\documentclass[english, 11pt]{article}

\usepackage{custom_tex}
\usepackage{soul}
\usepackage{multirow}
\usepackage{tikz}
 
\allowdisplaybreaks
\usepackage[page]{appendix}

\newcommand{\titlename}{Minimax Rates and Adaptivity in Combining Experimental and Observational Data}

\title{\titlename}

\newcommand{\rct}{c}

\newcommand{\Rct}{\calC}
\newcommand{\obs}{o}
\newcommand{\Obs}{\calO}
\newcommand{\weight}{\omega}
\newcommand{\amse}{{\mathsf{amse}}}

\newcommand{\cipw}{\hat\beta_{c, \mathsf{IPW}}}
\newcommand{\caipw}{\hat\beta_{c, \mathsf{AIPW}}}

\newcommand{\cippw}{\hat\beta_{c, \mathsf{IPPW}}}
\newcommand{\caippw}{\hat\beta_{c, \mathsf{AIPPW}}}

\newcommand{\ooipw}{\hat\beta_{o_2, \mathsf{IPW}}}
\newcommand{\ooaipw}{\hat\beta_{o_2, \mathsf{AIPW}}}

\newcommand{\len}{\mathsf{Len}}
\newcommand{\ci}{\mathsf{CI}}
\newcommand{\orac}{\mathsf{orac}}

\renewcommand{\widehat}{\hat}

\author[a]{Shuxiao Chen\thanks{Email:
\href{mailto:shuxiaoc@wharton.upenn.edu}{shuxiaoc@wharton.upenn.edu}}}
\author[a]{Bo Zhang\thanks{Email: \href{bozhan@wharton.upenn.edu}{bozhan@wharton.upenn.edu}}}
\author[b]{Ting Ye\thanks{Email: \href{mailto:tingye1@uw.edu}{tingye1@uw.edu}}}
\affil[a]{\textit{Department of Statistics and Data Science, The Wharton School, University of Pennsylvania}}
\affil[b]{\textit{Department of Biostatistics, University of Washington}}
\date{\today}
\begin{document}

\maketitle

\begin{abstract}

\noindent 
Randomized controlled trials (RCTs) are the gold standard \textcolor{black}{for evaluating} the causal effect of a treatment; however, they often have limited sample sizes and sometimes poor generalizability. 
On the other hand, non-randomized, observational data derived from large administrative databases have massive sample sizes and better generalizability, but they are prone to unmeasured confounding bias. 
It is thus of considerable interest to reconcile effect estimates obtained from randomized controlled trials and observational studies investigating the same intervention, potentially harvesting the best from both realms. 
In this paper,  we theoretically characterize the potential efficiency gain of integrating observational data into the RCT-based analysis from a minimax point of view.
For estimation, we derive the minimax rate of convergence for the mean squared error, and propose a fully adaptive \emph{anchored thresholding} estimator that attains the optimal rate up to poly-log factors.
For inference, we characterize the minimax rate for the length of confidence intervals and show that adaptation (to unknown confounding bias) is in general impossible. 
A curious phenomenon thus emerges: for estimation, the efficiency gain from data integration can be achieved without prior knowledge on the magnitude of the confounding bias; for inference, the same task becomes information-theoretically impossible in general.
We corroborate our theoretical findings 
using simulations and a real data example from the RCT DUPLICATE initiative \citep{franklin2021emulating}.
\end{abstract}

\smallskip

\noindent\textit{Keywords:} causal inference; generalizability; integrative data analysis; observational studies; randomized controlled trials; transfer learning.

\addtocontents{toc}{\protect\setcounter{tocdepth}{0}}

\section{Introduction}
Evaluating the causal effect of a medical treatment or policy intervention remains the central query in biomedical and social sciences. While the randomized controlled trials (RCTs) remain the gold standard \textcolor{black}{for generating high-quality causal evidence}, they often have limited sample sizes and are not representative of the target population of scientific interest \citep{rothwell2005external, Stuart:2011aa, deaton2009instruments,deaton2018understanding}. 
On the other hand, a plethora of observational data have been collected for administrative and research purposes and are increasingly available to researchers in the form of disease registries, biobanks, administrative claims databases and electronic health records. Observational data typically have massive sample sizes and better resemble the population of scientific interest; however, researchers often refrain from interpreting the obtained conclusion as \emph{causal} due to the almost universal concern of unmeasured confounding bias. In fact, according to authors instructions of the flagship {\it Journal of the American Medical Association} (\citealp{JAMA_instructions}), ``causal language (including use of terms such as effect and efficacy) should be used only for randomized clinical trials.'' The goal of this article is to study how to estimate the average treatment effect (ATE) of a well-defined target population by reconciling RCTs and observational studies investigating the same scientific question and potentially harvesting the best from both realms.

One important special case of our framework is estimating ATE of a target population using only relevant RCTs; in other words, we would like to generalize the RCT estimates to a target population of scientific interest. This topic has been extensively studied in the literature under the themes of ``generalizability", ``transportability'', and ``improving RCTs' external validity'' \citep{cole2010generalizing, Stuart:2011aa,tipton2013improving, pearl2014external,hartman2015sample, bareinboim2016causal,Buchanan:2018aa,Dahabreh:2019aa}; see review papers by \citet{colnet2020causal} and \citet{degtiar2021review} and references therein. 
In this special case, observational data are only used to inform researchers of subject characteristics in the target population; no treatment or outcome observational data are needed or leveraged in research along this line. 
Another important special case is when the target population concurs with the RCT study population and observational data are used to improve the efficiency of RCT estimates. 
To this end, \citet{gagnon2021precise} proposed a design-based approach that leverages large-scale observational data to fit regression or flexible machine learning models and conduct regression adjustment for randomization-based inference of RCT data; see \citet{aronow2013class} for similar strategies. 
Observational data are not used beyond helping gauge a better regression adjustment fit to improve the randomization inference efficiency and are thus not leveraged to their full potential.

Estimating causal effects using both RCTs and observational studies inevitably involves the century-old bias-variance trade-off. 
Observational data play the role of double-edged sword here: they are blessed with a sometimes enormous sample size and could dramatically improve estimation precision under suitable conditions; 
on the other hand, they may leak bias into the otherwise unbiased RCT estimates. 
Developing a principled and adaptive approach to bias-variance trade-off thus lies at the heart of combining RCTs and observational studies for optimal ATE estimation. 
To the best of our knowledge, few works have formally studied this problem, let alone the development of adaptive algorithms.
Two notable exceptions are \citet{yang2020elastic} and \citet{yang2020improved}. The former developed a hypothesis-testing-based integrative estimator that pools the RCT data and observational data when a certain criterion is met. The bias-variance trade-off is controlled by the choice of the critical region when testing whether the observational data violate the no unmeasured confounding assumption. 
The latter posited a functional form on the unmeasured confounding bias and proposed to estimate it using semiparametric theory. 
However, both works rely on the assumption that the sample size ratio between the observational data and the RCT data is \emph{bounded}. Such an assumption does not capture the fact that the sample size of observational data can be \emph{much larger} than that of the RCT data, and it effectively rules out the possibility of obtaining \emph{order-wise} improvement, even if the confounding bias is negligible.


Our overarching goal is to theoretically characterize the potential efficiency gain from integrating observational data into the analysis based on RCT alone. 
The main contributions are two-folds. 
\begin{enumerate}
\item For estimation, we derive the minimax rate of convergence for the mean squared error. Interestingly, the rate can be achieved by an oracle estimator that takes a \emph{dichotomous} form: when the magnitude of the confounding bias exceeds a certain threshold, we estimate the causal effect using the RCT alone; when the magnitude is below this threshold, we pool the two data sources together, as if the confounding bias never exists.
We further propose a fully adaptive estimator, termed \emph{anchored thresholding} estimator, that nearly attains the minimax rate up to poly-log factors without any prior knowledge on the magnitude of the bias. 

\item For statistical inference, we characterize the minimax {rate for the} length of confidence intervals (CIs).  {The minimax rate for CI length is similar in functional form to that for estimation, and can be attained similarly by a dichotomous oracle procedure.}
However, the story for adaptivity becomes drastically different.
We show that constructing adaptive CIs that attains the minimax rate is generally impossible, unless there is strong prior knowledge on the magnitude of the confounding bias.
As a corollary, our results indicate the existence of a threshold such that 
if one does not believe the magnitude of confounding bias is much smaller than this threshold, then \emph{any} uniformly valid CIs can only outperform the RCT-only CI by at most a constant factor.
\end{enumerate}


\noindent\textbf{Notations.}
We conclude this section by introducing some notations that will be used throughout this paper. 
Given $a, b \in \bbR$, we denote $a\lor b = \max\{a, b\}$ and $a\land b = \min\{a, b\}$.
For two positive sequences $\{a_n\}$ and $\{b_n\}$, we write 
$a_n \lesssim b_n$ to denote $\limsup a_n/b_n < \infty$, 
and we let $a_n \gtrsim b_n$ to denote $b_n \lesssim a_n$. 
Meanwhile, the notation $a_n\asymp b_n$ means $a_n \lesssim b_n$ and $a_n \gtrsim b_n$ simultaneously.
Moreover, we write $a_n\ll b_n$ or $a_n = o(b_n)$ to mean $b_n/a_n\to\infty$ and $a_n\gg b_n$ to mean $b_n \ll a_n$.
We use $\overset{p}{\to}$ to denote convergence in probability and we use $X_n = o_p(1)$ to denote $X_n\overset{p}{\to} 0$.
\section{Set-up}
\label{sec: setup}
\subsection{Estimands and Baseline Estimators}
Suppose the existence of a binary treatment $ A\in \{0,1\} $, potential outcomes $\{Y(1), Y(0)\}$ and baseline observed covariates $\mathbf{X}$. 
Throughout the article, we assume the consistency and Stable Unit Treatment Value Assumption (SUTVA)  so that the observed outcome $Y$ satisfies $Y= AY(1) + (1-A)Y(0)$ \citep{Rubin1980}.
Our estimand of interest is the average treatment effect of a target population (PATE) $\beta^\star = \mathbb{E}[Y{(1)} - Y{(0)}]$, where $\mathbb{E}[\cdot]$ is taken with respect to the joint distribution of $(\mathbf{X}, Y(0), Y(1))$ over the target population.

A random sample of size $n_c$ from the target population was selected into an RCT according to some sample selection protocol. We use $S = c$ to indicate selection into the RCT and let $\mathcal{C}$ denote the associated index set so that $|\mathcal{C}| = n_c$. We write the RCT data as $\mathcal{D}_{c} = \{(Y_i, A_i, \mathbf{X}_i, S_i = c):~i \in \mathcal{C}\}$, which are  assumed to be independent and identically distributed according to the joint law of $\big(Y(1), Y(0), A, \mathbf{X}\big) \mid S = c$. We denote the average treatment effect in the RCT population as $\beta^\star_{c} = \mathbb{E}[Y(1) - Y(0) \mid S = c]$. 
\begin{assumption}[Design of RCT]\label{assump: design of RCT} \rm
Assume $A \indep (Y(1), Y(0)) \mid \mathbf{X}, S=c$, the treatment assignment probability $ \pi_c(\mathbf{x}):= P(A=1\mid \mathbf{X}=\mathbf{x}, S=c) $ is \emph{known}, and $\pi_c(\mathbf{X})>0$ with probability 1 with respect to the law of $\bfX\mid S = c$. 
\end{assumption}

Under Assumption \ref{assump: design of RCT},  $\beta^\star_c$ is identified from  $\mathcal{D}_c$ and can be estimated via the inverse probability weighting (IPW) estimator  \citep{horvitz1952generalization}:
\begin{equation}
\label{eq:cipw}
\hat\beta_{c, \textsf{IPW}}  = \frac{1}{n_c} \sum_{i\in \mathcal{C}}  \bigg(\frac{A_i Y_i }{\pi_c(\mathbf{X}_i)} - \frac{(1-A_i) Y_i }{1- \pi_c(\mathbf{X}_i)}\bigg).
\end{equation}
Alternatively, $\beta^\star_c$ can be estimated via the following augmented inverse probability weighting (AIPW) estimator \citep{robins1994estimation, bang2005doubly}:
\begin{align*}
\hat\beta_{c, \textsf{AIPW}}
&  = \frac{1}{n_c} \sum_{i\in \mathcal{C}}\bigg(\frac{A_i (Y_i -\mu_{c1} (\mathbf{X}_i; \hat \bsalpha_{c1} ) )  }{\pi_c(\mathbf{X}_i)} - \frac{(1-A_i)(Y_i -\mu_{c0} (\mathbf{X}_i; \hat \bsalpha_{c0} ) ) }{1- \pi_c(\mathbf{X}_i)} \\
\label{eq:caipw}
& \qquad \qquad \qquad +\mu_{c1} (\mathbf{X}_i; \hat \bsalpha_{c1} )   - \mu_{c0} (\mathbf{X}_i; \hat \bsalpha_{c0} )\bigg),\numberthis
\end{align*}
where  $ \mu_{ca}(\bfx; \bsalpha_{ca} )$ is a parametric specification (parametrized by $\bsalpha_{ca}$) of the conditional outcome mean $\mu_{ca}(\bfx) = \mathbb{E}[Y \mid \bfX=\bfx, A = a, S = c]$ and $ \hat \bsalpha_{ca} $ an estimator of $ \bsalpha_{ca}$, for $a=0,1$. 
In an RCT with known treatment assignment probability $\pi_c(\bfx)$, the AIPW estimator $ \hat\beta_{c, \textsf{AIPW}}$ consistently estimates $\beta^\star_c$ even when $\mu_{ca} (\bfx; \bsalpha_{ca})$ do not correctly specify $\mu_{ca}(\bfx)$, and is semiparametric efficient if they do \citep{bang2005doubly,tsiatis2007semiparametric}. 
The AIPW estimator $ \hat\beta_{c, \textsf{AIPW}}$ is also guaranteed to be more efficient than $\hat\beta_{c, \textsf{IPW}}$ when $\pi_c(\bfx)$ does not depend on the covariate vector $\bfx$ (e.g., under random treatment assignment),  $\mu_{ca}(\bfx; \bsalpha_{ca})$ are linear parametric models of $\mu_{ca}(\bfx)$, and $ \hat \bsalpha_{ca}$ are the corresponding least squares estimators \citep{tsiatis2008covariate, ye2020principles}. 

In parallel, a random sample of size $n_o$ from the target population was collected in an observational study. We use $S = o$ to indicate selection into the observational study and let $\mathcal{O}$ denote the associated index set so that $|\mathcal{O}| = n_o$. Typically, $n_o \gg n_c$.  We write the observational data as
$\mathcal{D}_{o} = \{(Y_i, A_i, \mathbf{X}_i, S_i = o):~i \in \mathcal{O}\}$, which are assumed to be independent and identically distributed according to the joint law of $\big(Y(1), Y(0), A, \mathbf{X}\big) \mid S = o$. We denote the  average treatment effect in the observational study population as $\beta^\star_{o} = \mathbb{E}[Y(1) - Y(0) \mid S = o]$. Without additional assumptions, $\beta^\star_\rct \neq \beta^\star_\obs$ and neither agrees with the estimand of interest $\beta^\star$. 
    
For estimation, we impose a positivity assumption on the treatment assignment probability in the observational data  \citep{rosenbaum1983central}. 

\begin{assumption}[Positivity of Treatment Assignment in Observational Study]\label{assump: positivity obs} 
Assume the treatment assignment probability $\pi_o(\mathbf{X}):= P(A=1\mid \bfX, S=o)>0$ with probability 1 with respect to the law of $\bfX \mid S = \obs$.
\end{assumption}
Let  $\mu_{oa}(\bfx) = \mathbb{E}[Y \mid \mathbf{X}=\bfx, A = a, S = o]$ be the conditional outcome means in the observational data, and  $\mu_{oa}(\bfx;  \bsalpha_{oa})$ be their corresponding parametric specifications for $a = 0, 1$. Also let $\pi_{o}(\bfx; \bseta)$ be a parametric specification of $\pi_o(\bfX)$. Under Assumption \ref{assump: positivity obs},  the observational study counterparts of $\cipw$ and $\caipw$ defined in \eqref{eq:cipw} and \eqref{eq:caipw} can then be formed as
\begin{align}
\label{eqn:oipw}
    \hat\beta_{o, \textsf{IPW}} &= \frac{1}{n_o} \sum_{i\in \mathcal{O}}  \bigg(\frac{A_i Y_i }{\pi_{o}(\mathbf{X}_i; \hat \bseta)} - \frac{(1-A_i) Y_i }{1- \pi_{o}(\mathbf{X}_i; \hat \bseta) }\bigg) ,\\
        \hat    \beta_{o, \textsf{AIPW}} &=  \frac{1}{n_o} \sum_{i\in \mathcal{O}}\bigg(\frac{A_i (Y_i -\mu_{o1} (\mathbf{X}_i; \hat \bsalpha_{o1} ) )  }{\pi_{o}(\mathbf{X}_i; \hat\bseta )} - \frac{(1-A_i)(Y_i -\mu_{o0} (\mathbf{X}_i; \hat \bsalpha_{o0} ) ) }{1- \pi_{o}(\mathbf{X}_i;\hat\bseta)}  \nonumber\\
\label{eqn:oaipw}
        & \qquad \qquad\qquad + \mu_{o1} (\mathbf{X}_i; \hat \bsalpha_{o1} )   - \mu_{o0} (\mathbf{X}_i; \hat \bsalpha_{o0} )\bigg),
\end{align}
where $\hat\bseta$ and $\hat\bsalpha_{oa}$ are estimators of $\bseta$ and $\bsalpha_{oa}$, respectively. 

\subsection{Internal Validity Bias}
\label{subsec: internal validity}

Internal validity refers to an effect estimate being (asymptotically) unbiased for the causal effect. Under Assumption \ref{assump: design of RCT} which states that treatment assignment in the RCT population is manipulated by the experimenter and depends only on the observed covariates, both $\hat\beta_{c, \textsf{IPW}}$ and $\hat\beta_{c, \textsf{AIPW}}$ are internally valid: they are consistent and asymptotically normal estimators for the causal estimand $\beta^{\star}_c$.  On the other hand, under  Assumption \ref{assump: positivity obs}, $\hat\beta_{o, \textsf{IPW}}$ and $\hat\beta_{o, \textsf{AIPW}}$ have internal validity bias as they are estimating a non-causal quantity $
\beta_o=\mathbb{E}\left[Y\mid  S = o, A=1\right] - \mathbb{E}\left[Y\mid  S = o, A=0\right]
$, which is not necessarily equal to $\beta^\star_o$ due to unmeasured confounding; see Appendix \ref{appx:internal_bias} for a more detailed analysis of the internal validity bias.

\subsection{External Validity Bias}
\label{subsec: external validity}
External validity is concerned with how well the effect estimates generalize to other contexts \citep{degtiar2021review}. This bias often arises from the interplay between effect heterogeneity and the difference between the study and target populations. We focus primarily on the setting where the observational data represent the target population, as formalized below.
\begin{assumption}[Target Population]\label{assump:target}
   The target distribution $(Y(1), Y(0), \bfX)$ and the observational study population $(Y(1), Y(0), \bfX)\mid S=o$ are the same.
\end{assumption}

We remark that our method and theory developed in Sections \ref{sec:estimation} and \ref{sec:inference} can be extended in a \emph{mutatis mutandis} fashion to accommodate other settings, such as when the RCT data represent the target population, or even when there is a separate target population that is distinct from either RCT or observational data. 


It is clear that $\beta^\star = \beta^\star_\obs$ under Assumption \ref{assump:target}. When the observational data are representative of the target population, the two RCT-based estimators $\cipw$ and $\caipw$, albeit internally valid, may suffer from external validity bias. To eliminate the external validity bias, we make the following standard \emph{conditional mean exchangeability} assumption, also known as the \emph{generalizability} assumption \citep{Stuart:2011aa, Dahabreh:2019aa, degtiar2021review}. 

\begin{assumption}[Conditional Mean Exchangeability]\label{assump: study selection exchangeability} 
$\mathbb{E}[Y(1)- Y(0)| \mathbf{X}, S = c] = \mathbb{E}[Y(1)- Y(0) | \mathbf{X}]$.
\end{assumption}

Assumption \ref{assump: study selection exchangeability} states that participating in the RCT does not modify the treatment effect in every strata defined by $\mathbf{X}$, which will hold if all patient characteristics that modify the treatment effect and differ between the RCT population and the target population are measured.  

Under Assumptions \ref{assump:target}, \ref{assump: study selection exchangeability} and a often-made \emph{RCT participation positivity} assumption which states that the participation probability defined as $e_\rct(\bfx) = \bbP(S = \rct\mid \bfX =\bfx, S \in \{\rct, \obs\})$ is strictly bounded away from $0$ and $1$ almost surely, the PATE  $\beta^\star$ can be identified from the RCT data via inverse probability of participation weighting (IPPW) or augmented inverse probability of participation weighting (AIPPW) \citep{Stuart:2011aa,Dahabreh:2019aa,lu2019causal, degtiar2021review}. 
Here we use $S \in \{\rct, \obs\}$ to indicate that the probability is calculated for a subject in the RCT or the observational study. 
However, in practice, the  RCT participation probability would decay to zero if $n_\obs\gg n_\rct$ as both $n_o$ and $n_c$ go to infinity, thus invalidating the  RCT participation positivity assumption.  For this reason, we make a weaker assumption that does not assume away  decaying RCT participation probability.

\begin{assumption}[Bounded Density Ratio]
Assume $\bbP(\mathbf{X}= \mathbf{x}\mid S=o) / \bbP( \mathbf{X}= \mathbf{x} \mid S=c)<\infty$ for all $\mathbf{x}$ with positive density in the target population. 
\end{assumption}

To tackle the challenge imposed by decaying RCT participation probability, we randomly sample a subset of size $n_c$ from the observational data, denoted as $\mathcal{D}_{\mathcal{O}_1}$, and use $\mathcal{D}_{\mathcal{C}}\cup\mathcal{D}_{\mathcal{O}_1} $ to estimate the RCT participation probability and then adjust for the external validity bias of the RCT-based estimators using IPPW or AIPPW as follows:
\begin{align*}
\begin{split}
    \hat\beta_{c, \textsf{IPPW}}  &= \frac{1}{n_c} \sum_{i\in \mathcal{C}} \bigg[ \frac{A_i Y_i }{\pi_c(\mathbf{X}_i)} - \frac{(1-A_i) Y_i }{1- \pi_c(\mathbf{X}_i)}\bigg] \frac{ 1-e_{{c}}(\mathbf{X}_i; \hat  \bsxi )  }{e_{{c}}(\mathbf{X}_i; \hat  \bsxi )    },\\
\hat\beta_{c, \textsf{AIPPW}}   & = \frac{1}{n_c} \sum_{i\in \mathcal{C}} \bigg[\frac{A_i (Y_i -\mu_{c1} (\mathbf{X}_i; \hat \bsalpha_{c1} ) )  }{\pi_c(\mathbf{X}_i)} - \frac{(1-A_i)(Y_i -\mu_{c0} (\mathbf{X}_i; \hat \bsalpha_{c0} ) ) }{1- \pi_c(\mathbf{X}_i)}  \bigg] \frac{ 1-e_{{c}}(\mathbf{X}_i; \hat  \bsxi )  }{e_{{c}}(\mathbf{X}_i; \hat  \bsxi )  } \\
&\qquad + \frac{1}{n_c} \sum_{i\in \mathcal{O}_1} \mu_{c1} (\mathbf{X}_i; \hat \bsalpha_{c1} )   - \mu_{c0} (\mathbf{X}_i; \hat \bsalpha_{c0} ),
\end{split}
\end{align*}
where $e_{{c}}(\mathbf{x};  \bsxi )$ is a parametric specification of $e_\rct(\bfx)$ and $\hat \bsxi$ an estimator of $ \bsxi$ obtained from $\mathcal{D}_{\mathcal{C}}\cup\mathcal{D}_{\mathcal{O}_1} $. 
It can be shown that when $e_{{c}}(\mathbf{X}_i;  \bsxi )$ is correct, the probability limit of $\hat\beta_{c, \textsf{IPPW}} $ is $\beta^\star$; when either  $e_{{c}}(\mathbf{X}_i;  \bsxi )$ or $ \{\mu_{c1} (\mathbf{X}_i;  \bsalpha_{c1} ),  \mu_{c0} (\mathbf{X}_i;  \bsalpha_{c0} )\}$ is correct,  the probability limit of $\hat\beta_{c, \textsf{AIPPW}} $ is $\beta^\star$ \citep{Dahabreh:2019aa}. 

The remaining observational data $\mathcal{D}_{\mathcal{O}_2} = \mathcal{D}_{\mathcal{O}}\setminus\mathcal{D}_{\mathcal{O}_1}$ of size $n_{o_2}$ are then used to construct estimators $\hat\beta_{o_2,\textsf{IPW}}$ and $\hat\beta_{o_2,\textsf{AIPW}}$ according to \eqref{eqn:oipw} and \eqref{eqn:oaipw}, which potentially suffer from the internal validity bias as discussed in Section \ref{subsec: internal validity}. In practice, IPW-type estimators like $\cippw, \caippw, \ooipw$, and $\ooaipw$ may be unstable when the (estimated) weights are close to zero, and researchers often stabilize the weights by normalization  \citep{robins2007comment}; we use $\hat\beta_{c, \textsf{sIPPW}}$, $\hat\beta_{c, \textsf{sAIPPW}}$, $\hat\beta_{o_2, \textsf{sIPW}}$, and $\hat\beta_{o_2, \textsf{sAIPW}}$ to denote their stabilized versions. More details are provided in Appendix \ref{appx:stablization}.

The consistency, asymptotic normality, and influence function representation of all aforementioned estimators are standard and can be immediately obtained from \citet[Chapter 6.1]{newey1994}; see Appendix \ref{appx:influence_func} for details.

\section{Estimation}\label{sec:estimation}

\subsection{An Oracle Estimator and Its Optimality}
\textcolor{black}{
Let $\hat \beta_\rct$ denote a generic estimator constructed from the RCT data and $\hat \beta_\obs$ a generic estimator constructed from the observational data.
For instance, if there is no external validity bias, then researchers could take $\hat\beta_\rct = \hat\beta_{c, \textsf{IPW}}$, or $\hat\beta_{\rct} = \hat\beta_{\rct, \textsf{AIPW}}$, or their stabilized versions, and $\hat\beta_\obs = \hat\beta_{o, \textsf{IPW}}$, or $\hat\beta_\obs = \hat\beta_{o, \textsf{AIPW}}$, or their stabilized versions.
In the presence of external validity bias, based on the discussion in Section \ref{subsec: external validity}, it is desirable to choose $\hat\beta_\rct = \hat\beta_{c, \textsf{IPPW}}$, or $\hat\beta_\rct = \hat\beta_{c, \textsf{AIPPW}}$, or their stabilized versions, and $\hat\beta_\obs = \hat\beta_{o_2, \textsf{IPW}}$, or $\hat\beta_\obs = \hat\beta_{o_2, \textsf{AIPW}}$, or their stabilized versions.}
Recall that $\hat\beta_\obs$ potentially suffers from the internal validity bias quantified by 
\begin{equation*}
\Delta =  \beta^\star - \beta_\obs.
\end{equation*}

We assume both $\hat\beta_\rct$ and $\hat\beta_\obs$ admit asymptotic linear expansions, formalized by Assumption \ref{assump:asymp_linear}.
\begin{assumption}[Existence of Asymptotic Linear Expansions]
\label{assump:asymp_linear}
Suppose there exist influence functions $\psi_\rct(\bfX, A, Y)$, $\psi_\obs(\bfX, A, Y)$ such that as $n_\rct \land n_{\obs} \to \infty$, 
\begin{align*}
    \sqrt{n_\rct}(\hat\beta_\rct - \beta^\star) &  = \frac{1}{\sqrt{n_\rct}} \sum_{i\in \Rct}  \psi_\rct(\bfX_i, A_i, Y_i) +  o_p(1),  \\
    \sqrt{n_{\obs}}(\hat\beta_\obs - \beta_\obs) &  = \frac{1}{\sqrt{n_{\obs}}} \sum_{i\in\Obs}  \psi_\obs(\bfX_i, A_i, Y_i) +  o_p(1).
\end{align*}
Moreover, the influence functions are mean zero with finite second and third moments given by $\bbE[(\psi_\rct(\bfX_i, A_i, Y_i))^2] = \sigma_\rct^2$ and $\bbE[|\psi_\rct(\bfX_i, A_i, Y_i)|^3] = \rho_\rct$ for all $i \in \calC$, and $\bbE[(\psi_\obs(\bfX_i, A_i, Y_i))^2] = \sigma_\obs^2$ and $\bbE[|\psi_\obs(\bfX_i, A_i, Y_i)|^3] = \rho_\obs$ for all $i \in \Obs$. 
\end{assumption}

\begin{remark}
The asymptotic linear expansion part of Assumption \ref{assump:asymp_linear} is satisfied by various choices of $\hat\beta_\rct$ and $\hat\beta_\obs$; see Section \ref{subsec: internal validity} and \ref{subsec: external validity}. 
Assumption \ref{assump:asymp_linear} also assumes the existence of finite third moments as our analysis makes use of the Berry-Essen theorem. While this assumption could be relaxed, e.g., by invoking generalizations of Berry-Esseen theorem that only requires finite fractional moment \citep[Chapter 16]{feller2008introduction}, we do not pursue it here.
\end{remark}
\begin{remark}
    \textcolor{black}{
    Careful readers may have noticed that the asymptotic linear expansion technically does not hold for settings when a part of the observational data is used in the construction of $\hat\beta_\rct$ (e.g., $\hat\beta_{\rct, \textsf{IPPW}}$ and $\hat\beta_{\rct, \textsf{AIPPW}}$). However, in such scenarios, we can simply overload the notations and re-define 
    $\calC \leftarrow \calC \cup \calO_1$, $n_\obs\leftarrow n_{\obs_2}$, and $\calO\leftarrow\calO_{2}$, and all the theoretical results in later sections will hold provided $n_\obs \geq (1+\ep)n_\rct$ for some absolute constant $\ep >0$.}
\end{remark}



\textcolor{black}{
Let us introduce a scalar $\overline\Delta \geq 0$ and assume that $|\Delta| \leq \overline\Delta$.
To motivate the procedure, consider the following two extreme cases.
If $\overline\Delta$ is very large, then in the worst case, incorporating the observational data can only do harm to our task of estimating $\beta^\star$. 
In this case, the worst-case optimal choice is arguably the RCT-only estimator $\hat\beta_\rct$. 
On the other hand, if $\overline{\Delta}$ is negligibly small, then the problem is reduced to aggregating two asymptotically (almost) unbiased and Gaussian estimators.
As the bias is negligibly small, a natural choice is taking the convex combination of the two estimators, with the weights chosen to minimize the overall variance:}
\begin{align}
    \label{eq:naive_estimator}
    \hat\beta_\weight = (1-\weight) \hat\beta_\rct + \weight \hat \beta_\obs,
    \qquad 
    \textnormal{where }
    \weight = \frac{\sigma_\rct^2/n_\rct}{ \sigma_\rct^2/n_\rct + \sigma_\obs^2/n_{\obs}}.
\end{align}
We refer to $\hat\beta_\weight$ as the \emph{naively-pooled} estimator.

\textcolor{black}{Following the above discussion, if we have prior knowledge on $\overline\Delta$, then we can choose which estimator to use based on the magnitude of $\overline\Delta$.}
In particular, we can implement the following oracle strategy: (1) when $\overline\Delta$ exceeds a certain threshold, then use the RCT-only estimator; (2) when $\overline\Delta$ is below that threshold, then use the naively-pooled estimator. The performance of such a procedure is given in the following theorem.

\begin{theorem}[Performance of the Oracle Estimator]
\label{thm:oracle_estimator}
For any estimator $\hat\beta$ of $\beta^\star$, we let $\amse(\hat\beta)$ be its asymptotic mean-squared error, so that 
$
    {\bbE[|\hat \beta - \beta^\star|^2]}/{\amse(\hat\beta)} \to 1
$
as $n_\rct \land n_{\obs}\to \infty$.
Under Assumption \ref{assump:asymp_linear}, as $n_\rct\land n_{\obs} \to \infty$, we have
\begin{align*}
    \amse(\hat\beta_\rct) = \frac{\sigma_\rct^2}{n_\rct},
    \qquad
    \amse(\hat\beta_\weight) =  \frac{\sigma_\rct^2 \sigma_\obs^2}{n_\rct \sigma_\obs^2 + n_{\obs} \sigma_\rct^2} + \bigg(\frac{\sigma_\rct^2/n_\rct}{\sigma_\rct^2/n_\rct + \sigma_\obs^2/n_{\obs}}\bigg)^2 \cdot \Delta^2.
\end{align*}
As a result, the following oracle estimator
\begin{align}
    \label{eq:oracle_estimator}
    \hat\beta_{\orac} = 
    \begin{cases}
        \hat\beta_\rct & \textnormal{ if } \overline\Delta \geq {\sigma_\rct}/{\sqrt{n_\rct}}\\
        \hat\beta_\weight & \textnormal{ otherwise}
    \end{cases}
\end{align}
satisfies 
\begin{align}
    \label{eq:amse_oracle_estimator}
    \amse(\hat\beta_\orac) \leq\frac{\sigma_\rct^2 \sigma_\obs^2}{n_\rct \sigma_\obs^2 + n_{\obs} \sigma_\rct^2} + \frac{\sigma_\rct^2}{n_\rct} \land \overline\Delta^2.
\end{align}
\end{theorem}
\begin{proof}
    See Appendix \ref{prf:thm:oracle_estimator}.
\end{proof}

The choice of setting the threshold at $\sigma_\rct/\sqrt{n_\rct}$ in the oracle estimator is based on the following calculations. Note that in practice, the observational data are usually much more abundant than the RCT data, i.e., ${\sigma_\obs^2}/{n_{\obs}} \ll \sigma_\rct^2/n_\rct$. In this case, we have $\weight = 1 + o(1)$, and ${\sigma_\rct^2 \sigma_\obs^2}/({n_\rct \sigma_\obs^2 + n_{\obs} \sigma_\rct^2}) \ll {\sigma_\rct^2}/{n_\rct}$. 
Thus, we can simplify the formula for $\amse(\hat\beta_\weight)$ as $o(\sigma_\rct^2/n_\rct) + (1+o(1))\Delta^2$.
It follows that the RCT-only estimator would outperform the naively-pooled estimator (in terms of $\amse$) provided the magnitude of the bias exceeds $\sigma_\rct/\sqrt{n_\rct}$.
Such a choice of the threshold is also intuitive from the perspective of the classical bias-variance decomposition: incorporating $\hat\beta_\obs$ will only be useful if its bias does not exceed the standard deviation of $\hat\beta_\rct$.

Curious readers may wonder if the upper bound in \eqref{eq:amse_oracle_estimator}, achieved by the oracle estimator, is information-theoretically optimal among all estimators.
In the following, we describe a simple yet canonical data-generating process and show that $\hat\beta_\orac$ is indeed minimax rate optimal. 
Apart from being a model of what happens in practice, this data generating process serves as a theoretical benchmark that enables us to give a rigorous treatment of optimality. 

\begin{model}
\label{model}
Suppose the RCT data are generated as follows. For each $i\in\Rct$, we first generate a covariate vector $\bfX_i\sim \bbP_{X}$. Then, we generate the treatment $A_i \sim \Bern(\pi_\rct(\bfX_i))$ for some function $\pi_\rct(\bfx) \in (0, 1)$. 
Finally, we generate $Y_i(A_i) \sim N((A_i-1/2)\beta^\star + \bsgamma_\rct^\top \bfX_i, \sigma_\rct^2)$ where $\bsgamma_\rct$ is a vector with the same dimension as $\bfX_i$. 

Meanwhile, suppose the observational data are generated as follows. 
For each $i\in\Obs$, we first generate the treatment $A_i\sim \Bern(\pi_\obs)$ for some scalar $\pi_\obs\in(0, 1)$.
Then, we generate an unmeasured confounding variable $U_i\sim N(-(A_i-1/2)\Delta, \sigma_\obs^2/2)$ and a covariate vector $\bfX_i\sim \bbP_{X}$.
We finally generate $Y_i(A_i)\sim N((A_i-1/2)\beta^\star + U_i + \bsgamma_\obs^\top \bfX_i, \sigma_\obs^2/2)$, where $\bsgamma_\obs$ is a vector with the same dimension as $\bfX_i$. 
\end{model}

By the design of RCT (i.e., Assumption \ref{assump: design of RCT}), the observed responses in the RCT data satisfy 
$$
Y_i\mid A_i, \bfX_i \overset{i.i.d.}{\sim} N((A_i-1/2)\beta^\star + \bsgamma_\rct^\top \bfX_i, \sigma_\rct^2).
$$ 
From our discussion in Section \ref{sec: setup}, it is easy to construct an estimator $\hat\beta_\rct$ that is asymptotically normal with mean $\beta^\star$ and variance of order $\sigma_\rct^2/n_\rct$ (note that we do not need to weight by the participation probability as $\beta^\star_\rct = \beta^\star_\obs = \beta^\star$ in the above model).
Meanwhile, marginalizing over $U_i$, the observed responses in the observational data satisfy 
$$
Y_i\mid A_i, \bfX_i \overset{i.i.d.}{\sim} N((A_i-1/2)\beta_\obs + \bsgamma_\obs^\top \bfX_i , \sigma_\obs^2),
$$ 
where $\beta_\obs = \beta^\star - \Delta$. Thus, even if $\bsgamma_\obs$ is known, the best one can do is to produce an estimator $\hat\beta_\obs$ that is centered at $\beta_\obs$. 

Consider the following parameter space:
\begin{align}
    \label{eq:param_space}
    \calP_{\overline\Delta} = \{(\beta^\star, \beta_\obs): |\beta^\star - \beta_\obs|\leq \overline\Delta\}.
\end{align}
The following theorem gives the minimax lower bound under Model \ref{model}.

\begin{theorem}[Minimax Lower Bound]
\label{thm:lb}
Under Model \ref{model}, we have
\begin{equation*}
    \inf_{\hat \beta} \sup_{(\beta^\star, \beta_\obs)\in\calP_{\overline\Delta}} \bbE[|\hat\beta - \beta^\star|^2] \gtrsim \frac{\sigma_\rct^2 \sigma_\obs^2}{n_\rct \sigma_\obs^2 + n_{\obs} \sigma_\rct^2} + \frac{\sigma_\rct^2}{n_\rct} \land \overline\Delta^2,
\end{equation*}
where the infimum is taken over all measurable functions of the observed RCT and observational data.
\end{theorem}
\begin{proof}
    See Appendix \ref{prf:thm:lb}.
\end{proof}

\begin{figure}[t]
    \centering
    \includegraphics[width = 0.49\textwidth]{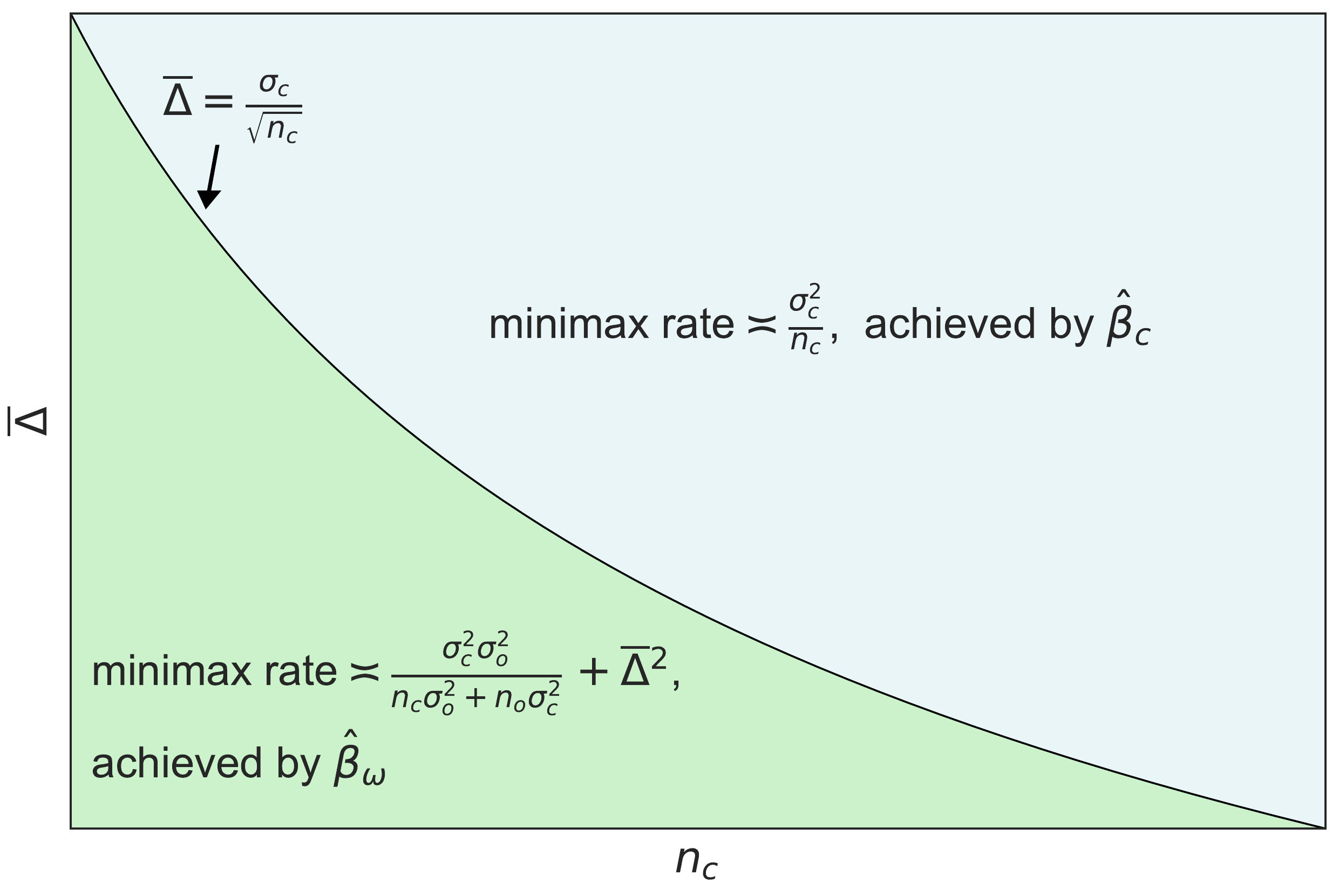}
    \caption{Phase transition for the minimax rate. When $\overline\Delta\gtrsim \sigma_\rct/\sqrt{n_\rct}$, the RCT-only estimator is minimax optimal; when $\overline\Delta\lesssim \sigma_\rct/\sqrt{n_\rct}$, the naively-pooled estimator is minimax optimal.} 
    \label{fig: phase_transition}
\end{figure}

The above theorem, along with Theorem \ref{thm:oracle_estimator}, confirms the minimax optimality of the oracle estimator $\hat\beta_\orac$. These two theorems together also reveal a curious phase transition phenomenon depicted in Figure \ref{fig: phase_transition}: when $\overline\Delta \gtrsim \sigma_\rct/\sqrt{n_\rct}$, then the minimax rate becomes $\Theta(\sigma_\rct^2/n_\rct)$, and such a rate is achieved by the RCT-only estimator $\hat\beta_\rct$;
when $\overline\Delta \lesssim \sigma_\rct/\sqrt{n_\rct}$, then the minimax rate becomes $\Theta({\sigma_\rct^2 \sigma_\obs^2}/({n_\rct \sigma_\obs^2 + n_{\obs} \sigma_\rct^2}) + \overline\Delta^2)$, and such a rate is achieved by the naively-pooled estimator $\hat\beta_\weight$.
A conceptually similar phase transition phenomenon has been identified in other statistical problems that involve multiple data sources \citep{chen2021theorem}.

\subsection{Adaptation via {Anchored} Thresholding}\label{subsec: adaptive_procedure}
In this section, we design a fully adaptive estimator that does not need prior knowledge of $\overline\Delta$, and we prove that this estimator achieves nearly the same performance as the oracle estimator $\hat\beta_{\orac}$, up to poly-log factors. 

If we can estimate the bias $\Delta$ sufficiently well by some $\hat\Delta$, then we can ``de-bias'' the observational-data-only estimator by computing $\hat\beta_\obs + \hat\Delta$. Since the resulting estimator is approximately unbiased, we may aggregate $\hat\beta_\rct$ and $\hat\beta_\obs + \hat\Delta$ by forming
$
    (1 - \weight) \hat\beta_\rct + \weight (\hat \beta_\obs + \hat\Delta),
$
similar to our construction of the naively-pooled estimator in \eqref{eq:naive_estimator}.

There are two key issues in the above construction. First, since the variances are usually unknown in practice, we need to replace the weight $\omega$ by its estimated counterpart
$$
    \hat \weight = \frac{\hat \sigma_\rct^2/n_\rct}{\hat \sigma_\rct^2/n_\rct + \sigma_\obs^2/n_{\obs}},
$$
where $\hat\sigma_\rct, \hat\sigma_\obs$ are estimators of $\sigma_\rct, \sigma_\obs$, respectively. The other, and perhaps more important issue, lies in the choice of $\hat\Delta$. Since $\hat\beta_\rct$ and $\hat\beta_\obs$ are respectively asymptotically unbiased for $\beta^\star$ and $\beta_\obs$, a tempting choice of $\hat\Delta$ is
$$
    \tilde\Delta = \hat\beta_\rct - \hat\beta_\obs.
$$
However, a careful thought reveals that $\tilde\Delta$ is not sufficient for our purpose of obtaining a better estimator than $\hat\beta_\rct$: if we use $\hat\Delta = \tilde \Delta$, then we would get
$$
    (1-\weight) \hat\beta_\rct + \weight (\hat\beta_\obs + \tilde\Delta) = (1-\weight) \hat\beta_\rct + \weight \hat\beta_\rct = \hat\beta_\rct.
$$
In other words, the accuracy of $\tilde\Delta$ is not enough to provide any efficiency gain against the RCT-only estimator $\hat\beta_\rct$. Clearly, more work is needed.


To proceed, we make two observations. First, we know that $\tilde \Delta$, after proper scaling, is asymptotically Gaussian as $n_\rct \land n_{\obs}\to\infty$, and the asymptotic mean and variance of $\tilde \Delta$ are given by $\Delta$ and ${\sigma_\rct^2}/{n_\rct} + {\sigma_\obs^2}/{n_{\obs}}$, respectively. A standard concentration of measure argument gives that $|\Delta - \tilde \Delta| \lesssim \lambda \sqrt{({\sigma_\rct^2}/{n_\rct} + {\sigma_\obs^2}/{n_{\obs}})}$ with probability tending to one, as long as $\lambda$ is a diverging sequence as $n_\rct\land n_{\obs}\to\infty$ \citep{boucheron2013concentration}.
Second, we know that the observational data 
are {useful only} when the bias $\Delta$ is small. These two observations motivate the following procedure:

\begin{equation}
    \label{eq:est_delta}
    \hat \Delta_\lambda = \argmin_{\delta} |\delta| \qquad \textnormal{subject to } |\delta - \tilde \Delta| \leq 
    \lambda \cdot 
    \sqrt{\frac{\hat\sigma_\rct^2}{n_\rct} + \frac{\hat\sigma_\obs^2}{n_{\obs}}},
\end{equation} 
where $\lambda > 0$ is a hyperparameter. Intuitively, the above procedure {selects an estimator of the bias $\Delta$} with the minimum absolute value (based on our second observation), among all ``plausible'' candidates, where the ``plausibility'' is quantified by our first observation.

Such a procedure is reminiscent of many well-known methods in high-dimensional statistics, such as the Dantzig selector in sparse linear models \citep{candes2007dantzig} and nuclear norm minimization in matrix completion \citep{candes2009exact,candes2010power}, where one seeks for an estimator that minimizes certain complexity measures (in our case, the absolute value) subject to a constraint that ensures the error control.

Note that \eqref{eq:est_delta} can be solved via soft-thresholding: 
$$
    \hat \Delta_\lambda = 
    \begin{cases}
        \textnormal{sgn}(\tilde \Delta)\Big(|\tilde \Delta| - \lambda \cdot
        \sqrt{\frac{\hat\sigma_\rct^2}{n_\rct} + \frac{\hat\sigma_\obs^2}{n_{\obs}}}\Big)
        & \textnormal{if } 
        |\tilde \Delta| \geq \lambda \cdot
        \sqrt{\frac{\hat\sigma_\rct^2}{n_\rct} + \frac{\hat\sigma_\obs^2}{n_{\obs}}},\\
        0
        & \textnormal{otherwise}.
    \end{cases}
$$
Thus, one can alternatively view \eqref{eq:est_delta} as a LASSO-type estimator \citep{tibshirani1996regression} {as it searches for the $\delta$} that minimizes $|\tilde \Delta - \delta|^2 + \lambda' |\delta|$ for a particular choice of $\lambda'$.

Once $\hat\Delta_\lambda$ is obtained, the final estimator for $\beta^\star$ is given by 
\begin{equation*}
    \hat \beta_\lambda = (1-\hat\weight) \hat\beta_\rct + \hat\weight (\hat\beta_\obs + \hat\Delta_\lambda).
\end{equation*}
{Because this estimator is anchored at the RCT-only estimator $\hat\beta_\rct$ and integrates a de-biased observational-data-only estimator obtained from soft-thresholding, we name $\hat\beta_\lambda$ the \emph{anchored thresholding} estimator.} 
The following theorem characterizes the performance of $\hat\beta_\lambda$.

\begin{theorem}[Performance of the Anchored Thresholding Estimator]
\label{thm:high_prob_bound_adaptive_estimator}
Let Assumption \ref{assump:asymp_linear} hold. In addition, assume $\sigma_\rct^2/n_\rct\gtrsim \sigma_\obs^2/n_{\obs}$
and $\hat\sigma_\rct \overset{p}{\to} \sigma_\rct, \hat\sigma_\obs \overset{p}{\to} \sigma_\obs$ as $n_\rct \land n_\obs \to \infty$. 
As long as we choose $\lambda \gtrsim 1$, there exists an absolute constant $\sfc > 0$ such that with probability at least $1- o(1) - e^{-\sfc \lambda^2}$, we have
\begin{equation*}
     |{\hat\beta_\lambda} - \beta^\star|^2 \lesssim  \lambda^2 \cdot \bigg(\frac{\sigma_\rct^2 \sigma_\obs^2}{n_\rct \sigma_\obs^2 + n_{\obs} \sigma_\rct^2} + \frac{\sigma_\rct^2}{n_\rct} \land \overline{\Delta}^2\bigg).
\end{equation*} 
\end{theorem}
\begin{proof}
    See Appendix \ref{prf:thm:high_prob_bound_adaptive_estimator}.
\end{proof}

By the above theorem, if we choose $\lambda \asymp \sqrt{\log (n_\rct\land n_{\obs})}$, then we conclude that with probability tending to one, the adaptive estimator $\hat\beta_\lambda$ achieves nearly the same performance as the oracle estimator, up to poly-log factors.
The $o(1)$ term in the high probability statement depends on the rate of convergence of $\hat\sigma_\rct$ and $\hat\sigma_\obs$, as well as the two $o_p(1)$ terms appeared in the asymptotic linear expansion of $\hat\beta_\rct$ and $\hat\beta_\obs$.

\section{Statistical Inference}\label{sec:inference}

\subsection{Minimax Rate for the Length of Confidence Intervals}
Under Assumption \ref{assump:asymp_linear}, we can always construct a valid CI for $\beta^\star$ using the RCT-only estimator $\hat\beta_\rct$, and its length is of order $\sigma_\rct/\sqrt{n_\rct}$.
If $n_\obs \gg n_\rct$, is it possible to construct a CI that is \emph{order-wise} better than the naive construction?
To answer this question in a rigorous fashion, we adopt the minimax framework pioneered by \cite{cai2004adaptation} and \cite{cai2017confidence}. 
We again use Model \ref{model} as our working model to benchmark the information-theoretic limit of CI construction.

Recall the parameter $\calP_{\overline\Delta}$ specified in \eqref{eq:param_space}. Let $\calI_\alpha(\calP_{\overline{\Delta}})$ be the set of all level $(1-\alpha)$ CIs that are uniformly valid over $\calP_{\overline{\Delta}}$:
$$
    \calI_\alpha(\calP_{\overline\Delta}) = \bigg\{[L(\calD_\Rct, \calD_\Obs)), U(\calD_\Rct, \calD_\Obs) ]: \bbP_{(\beta^\star, \beta_\obs)}(L \leq \beta^\star \leq  U)\geq 1-\alpha , \forall (\beta^\star, \beta_\obs) \in \calP_{\overline{\Delta}}\bigg\}.
$$
In the above display, $L$ and $U$ are two measurable functions of both the RCT data $\calD_\Rct$ and the observational data $\calD_\Obs$.
For any $\ci_\alpha = [L, U]\in \calI_\alpha(\calP_{\overline\Delta})$, its length is 
$
    \len(\ci_\alpha) = U - L.
$
The \emph{minimax length} of CIs over $\calP_{\overline\Delta}$ is defined as
$$
    \len_\alpha^\star(\calP_{\overline\Delta}) = \inf_{\ci_\alpha\in \calI_\alpha(\calP_{\overline\Delta})} \sup_{(\beta^\star, \beta_\obs)\in\calP_{\overline\Delta}} \bbE[\len(\ci_\alpha)].
$$
Analogous to the notion of minimax rate {for the mean squared error}, the minimax length captures the worst-case performance of CI construction over a pre-specified parameter space.

The following theorem characterizes the rate of convergence for the minimax length.
\begin{theorem}[{Minimax Rate of CI Length}]
\label{thm:minimax_length}
Under Model \ref{model}, we have
\begin{equation}
    \label{eq:minimax_length}
    \len_\alpha^\star(\calP_{\overline\Delta}) \gtrsim 
    \frac{\sigma_\rct\sigma_\obs}{\sqrt{n_\rct\sigma_\obs^2 + n_\obs\sigma_\rct^2}} + \frac{\sigma_\rct}{\sqrt{n_\rct}} \land \overline\Delta,
\end{equation} 
provided $\alpha \leq \Phi(-1/4)-\ep$ for some absolute constant $\ep\in(0, \Phi(-1/4))$.
\end{theorem}
\begin{proof}
    This is a special case of Theorem \ref{thm:minimax_adaptive_length} proved in Appendix \ref{prf:thm:minimax_adaptive_length}.
\end{proof}

The above result should not come as a surprise, as the right-hand side of \eqref{eq:minimax_length} is precisely the square-root of the minimax rate for the mean squared error given in Theorems \ref{thm:oracle_estimator} and \ref{thm:lb}.

Similar to the oracle estimator introduced in \eqref{eq:oracle_estimator}, we can construct an oracle CI that achieves the above lower bound. In particular, we define 
\begin{align}
    \label{eq:oracle_ci_L}
    L_\alpha & = 
    \begin{cases}
        \hat\beta_\rct - \Phi^{-1}(1-\frac{\alpha}{2}) \cdot \frac{\sigma_\rct}{\sqrt{n_\rct}}
        & \textnormal{ if } \overline\Delta \geq \sigma_\rct/\sqrt{n_\rct}\\
        \hat\beta_\weight - \Phi^{-1}(1-\frac{\alpha}{2})\cdot \frac{\sigma_\rct\sigma_\obs}{\sqrt{n_\rct\sigma_\obs^2 + n_{\obs}\sigma_\rct^2}} -\overline\Delta
        & \textnormal{ otherwise},
    \end{cases}\\
    \label{eq:oracle_ci_U}
    U_\alpha & = 
    \begin{cases}
        \hat\beta_\rct + \Phi^{-1}(1-\frac{\alpha}{2}) \cdot \frac{\sigma_\rct}{\sqrt{n_\rct}}
        & \textnormal{ if } \overline\Delta \geq \sigma_\rct/\sqrt{n_\rct}\\
        \hat\beta_\weight + \Phi^{-1}(1-\frac{\alpha}{2})\cdot \frac{\sigma_\rct\sigma_\obs}{\sqrt{n_\rct\sigma_\obs^2 + n_{\obs}\sigma_\rct^2}} + \overline\Delta
        & \textnormal{ otherwise}.
    \end{cases}
\end{align}
That is, when $\overline\Delta \geq \sigma_\rct/\sqrt{n_\rct}$, we construct the CI based on the RCT-only estimator $\hat\beta_\rct$, whereas when $\overline\Delta < \sigma_\rct/\sqrt{n_\rct}$, we first construct the CI based on the naively-pooled estimator and then enlarge it by $\overline\Delta$. The properties of this CI are given below.

\begin{theorem}[Oracle CI]
\label{thm:oracle_ci}
Consider the CI defined in \eqref{eq:oracle_ci_L} and \eqref{eq:oracle_ci_U}.
Under Assumption \ref{assump:asymp_linear}, for any $\alpha\in(0, 1)$, we have
$$
    \bbP(\beta^\star \in [L_\alpha, U_\alpha]) \geq 1-\alpha - o(1)
$$
as $n_\rct \land n_{\obs}\to \infty$. 
If 
$\alpha \geq \ep$ for some absolute constant $\ep\in(0, 1)$, then 
$$
\bbE[\len([L_\alpha, U_\alpha])]\lesssim \frac{\sigma_\rct\sigma_\obs}{\sqrt{n_\rct\sigma_\obs^2 + n_\obs\sigma_\rct^2}} + \frac{\sigma_\rct}{\sqrt{n_\rct}} \land \overline\Delta.
$$
\end{theorem}
\begin{proof}
    See Appendix \ref{prf:thm:oracle_ci}. 
\end{proof}

Combining \textcolor{black}{Theorems \ref{thm:minimax_length} and \ref{thm:oracle_ci}}, we conclude that when $\overline\Delta$ is known, it is indeed possible to construct a CI that is strictly better than the naive one based on the RCT-only estimator, and such an oracle CI is minimax optimal in terms of its expected length.

\subsection{When Is Adaptation Possible?}
Up to now, the results regarding inference is nearly identical to the results regarding estimation: as long as $\overline\Delta$ is known, we have an oracle procedure that attains the minimax lower bound. However, we are to see that the story for adaptation is drastically different. Without prior knowledge of $\overline\Delta$, it is in general not possible to construct a CI that achieves comparable performance as the oracle CI.

To formalize the above claim, let us introduce another parameter $\underline\Delta \leq \overline\Delta$ and consider a smaller parameter space
$$
    \calP_{\underline\Delta} = \{(\beta^\star, \beta_\obs): |\beta^\star - \beta_\obs|\leq \underline\Delta\} \subseteq \calP_{\overline\Delta}.
$$
We then define
$$
    \len_\alpha^\star(\calP_{\underline\Delta}, \calP_{\overline\Delta}) = 
    \inf_{\ci_\alpha\in \calI_\alpha(\calP_{\overline\Delta})} \sup_{(\beta^\star, \beta_\obs)\in\calP_{\underline\Delta}} \bbE_{(\beta^\star, \beta_\obs)} [\len(\ci_\alpha)],
$$
i.e., the minimax CI length over the smaller parameter space $\calP_{\underline\Delta}$ when we require the CI to be \emph{uniformly valid} over the larger parameter space $\calP_{\overline\Delta}$.
The above quantity measures the difficulty of adaptation to unknown $\underline\Delta$. 
For example, if 
$$
    \len_\alpha^\star(\calP_{\underline\Delta}, \calP_{\overline\Delta}) \asymp \len_\alpha^\star(\calP_{\underline\Delta}, \calP_{\underline\Delta}) = \len_\alpha^\star(\calP_{\underline\Delta}),
$$
then it tells that even if we require uniform coverage over the larger parameter space $\calP_{\overline{\Delta}}$, we can still construct a CI whose minimax length is of the same order as if we only require coverage over the smaller parameter space $\calP_{\underline{\Delta}}$.
In this case, we conclude that adaptation is possible.
On the contrary, if
$$
    \len_\alpha^\star(\calP_{\underline\Delta}, \calP_{\overline\Delta}) \gg  \len_\alpha^\star(\calP_{\underline\Delta}),
$$
then we know that requiring uniform coverage over $\calP_{\overline\Delta}$ would result in a much wider CI, and hence adaptation is impossible.

The following theorem gives a lower bound of $\len_\alpha^\star(\calP_{\underline\Delta}, \calP_{\overline\Delta})$.
\begin{theorem}[Minimax Length of Adaptive CIs]
\label{thm:minimax_adaptive_length}
Under Model \ref{model}, for any $\underline\Delta\leq \overline\Delta$, we have
\begin{equation}
    \label{eq:minimax_adaptive_length}
    \len_\alpha^\star(\calP_{\underline\Delta}, \calP_{\overline\Delta}) \gtrsim 
    \frac{\sigma_\rct\sigma_\obs}{\sqrt{n_\rct\sigma_\obs^2 + n_\obs\sigma_\rct^2}} + \frac{\sigma_\rct}{\sqrt{n_\rct}} \land \overline\Delta,
\end{equation} 
provided $\alpha \leq \Phi(-1/4)-\ep$ for some absolute constant $\ep\in(0, \Phi(-1/4))$.
\end{theorem}
\begin{proof}
    See Appendix \ref{prf:thm:minimax_adaptive_length}.
\end{proof}

Note that the lower bound in \eqref{eq:minimax_adaptive_length} has no dependence on $\underline\Delta$. Combining the above result with Theorems \ref{thm:minimax_length} and \ref{thm:oracle_ci}, we conclude that adaptation is in general not possible. As a concrete example, suppose that 
$$
    \frac{\sigma_\rct\sigma_\obs}{\sqrt{n_\rct\sigma_\obs^2 + n_\obs\sigma_\rct^2}} \lesssim \underline\Delta \ll \overline\Delta \lesssim \frac{\sigma_\rct}{\sqrt{n_\rct}},
$$
under which case we have
$$
    \len_\alpha^\star(\calP_{\underline\Delta}, \calP_{\overline\Delta}) \gtrsim \overline\Delta \gg \underline\Delta \asymp 
     \len_\alpha^\star(\calP_{\underline\Delta}).
$$
Moreover, the gap for adaptation, measured by the ratio between $\overline\Delta$ and $\underline\Delta$, can be as large as $\sqrt{n_\obs/n_\rct}$.

There are certain cases where adaptation is possible. For example, if we have prior knowledge that 
$
    \overline\Delta \lesssim {\sigma_\rct\sigma_\obs}/{\sqrt{n_\rct\sigma_\obs^2 + n_\obs\sigma_\rct^2}},
$
then the lower bound \eqref{eq:minimax_adaptive_length} becomes ${\sigma_\rct\sigma_\obs}/{\sqrt{n_\rct\sigma_\obs^2 + n_\obs\sigma_\rct^2}}$. Such a length can be achieved by constructing a CI based on the naively-pooled estimator and enlarging it by a constant multiple of ${\sigma_\rct\sigma_\obs}/{\sqrt{n_\rct\sigma_\obs^2 + n_\obs\sigma_\rct^2}}$. As another example, if $\overline\Delta\gtrsim \sigma_\rct/\sqrt{n_\rct}$, then the lower bound \eqref{eq:minimax_adaptive_length} reads $\sigma_\rct/\sqrt{n_\rct}$, and such a length is achieved by the CI based on the RCT-only estimator.

In view of Theorem \ref{thm:minimax_adaptive_length}, in practice, if we do not have strong prior knowledge that $\overline\Delta\ll \sigma_\rct/\sqrt{n_\rct}$, then the CI based on the RCT-only estimator is optimal from a worst-case point of view.


\section{Simulation Studies}
We consider the following data-generating process in the population:
\begin{equation} \label{eq: simu population}
    X_1, X_2, X_3 \sim N(0, 1);\qquad 
    Y(0) = X_1 + X_2 + X_3 + \epsilon,~Y(1) = Y(0) + \beta(X_1, X_2, X_3),
\end{equation}
where $\epsilon \sim N(0, 1)$. The first factor in the simulation design concerns about the effect heterogeneity:
\begin{description}
\item[Factor 1:] Treatment effect heterogeneity: a constant treatment effect $\beta(X_1, X_2, X_3) = 2$ and a heterogeneous treatment effect $\beta(X_1, X_2, X_3) = 2 - X_1 - X_2$. In either case, the PATE $\beta^\star = 2.$
\end{description}
We consider the following selection model into an RCT:
\begin{equation*}
\text{logit}\{\bbP(S = c \mid X_1, X_2, X_3)\} = -7 + X_1 - X_2.
\end{equation*}
According to this sample selection model, the marginal probability $P(S = c) \approx 0.25\%$. We simulate a population of $100,000$ so that $n_c \approx 250$ in each simulated RCT dataset. For each simulated RCT sample, we assign treatment $A$ randomly with probability $\pi_c = 1/2$ and $Y = A\cdot Y(1) + (1-A)\cdot Y(0)$. Note that the interplay between selection heterogeneity and treatment effect heterogeneity induces an external validity bias when $\beta(X_1, X_2, X_3) = 2 - X_1 + X_2$.

In parallel, we simulate an observational dataset with sample size $n_{\text{o}}$. Covariates and potential outcomes are simulated under \eqref{eq: simu population}. We assign treatment $A$ according to  $A \sim \text{Bern}(\text{expit}\{1 - X_1 - b \cdot X_3\})$, and $Y = A\cdot Y(1) + (1-A)\cdot Y(0)$. 
We consider a scenario where $X_3$ is an unmeasured confounder and hence the observed observational data $\mathcal{D}_o$ consist only of $\{X_1, X_2, A, Y\}$. This induces an internal validity bias when $b \neq 0$. The second and third factors are concerned about the observational data generating process. 

\begin{description}
\item[Factor 2:] Observational data sample size $n_o = 10,000$, $50,000$, and $100,000$.
\item[Factor 3:] Magnitude of internal validity bias controlled by $b$: (i) small or no bias $b = 0, 0.01, 0.1$, (ii) moderate bias $b = 0.5, 0.6, 0.7$, and (iii) large bias $b = 2, 3, 10$.
\end{description}

For each simulated pair of RCT and observational data, we construct an RCT-based estimator $\hat\beta_c$, an observational-data-based estimator $\hat\beta_o$,  an oracle estimator $\hat\beta_\orac$ that uses the oracle bias $\Delta$, and a family of adaptive estimators $\hat\beta_{\lambda}$ parameterized by $\lambda$. 
For a constant treatment effect $\beta(X_1, X_2, X_3) = 2$, we consider $(\hat\beta_c, \hat\beta_o) = (\hat\beta_{c, \textsf{IPW}}, \hat\beta_{o, \textsf{IPW}})$, $(\hat\beta_c, \hat\beta_o) = (\hat\beta_{c, \textsf{AIPW}}, \hat\beta_{o, \textsf{AIPW}})$, and we set $\lambda = \lambda_1 \cdot \sqrt{\log (n_\rct\land n_{\obs})}$ for some constant $\lambda_1$.  
For the heterogeneous treatment effect $\beta(X_1, X_2, X_3) = 2 - X_1 - X_2$, we consider $(\hat\beta_c, \hat\beta_o) = (\hat\beta_{c, \textsf{IPPW}}, \hat\beta_{o_2, \textsf{IPW}})$ and $(\hat\beta_c, \hat\beta_{o}) = (\hat\beta_{c, \textsf{AIPPW}}, \hat\beta_{o_2, \textsf{AIPW}})$. 
When constructing the $\hat\beta_{c, \textsf{IPPW}}$ and $\hat\beta_{c, \textsf{AIPPW}}$, we estimate the participation probability using RCT data and a random subsample of size $n_c$ from the observational data, and we set $\lambda = \lambda_1 \cdot \sqrt{\log (n_\rct\land n_{\obs_2})}$ for some constant $\lambda_1$.

\begin{table}[t]
\centering
\caption{Ratios of the mean squared error of each estimator to that of the oracle estimator $\hat{\beta}_\weight$ when $n_o = 10000$. Estimators considered in the table are the RCT-based estimator $\hat{\beta}_c$, the observational-data-based estimator $\hat{\beta}_o$, and the anchored thresholding estimator $\hat{\beta}_{\lambda}$. 
}
\label{tbl: simulation results}
\resizebox{\textwidth}{!}{
\begin{tabular}{ccccccccccc}
 &$\hat\beta_c$ & $\hat\beta_o$ &\multicolumn{2}{c}{$\hat\beta_{\lambda}$} & $\hat\beta_c$ & $\hat\beta_o$ &\multicolumn{2}{c}{$\hat\beta_{\lambda}$}\\
    \hline
    \hline
   \multicolumn{9}{c}{$\beta(X_1, X_2, X_3) = 2$} \\
    \hline
 \multirow{2}{*}{\begin{tabular}{c}b\end{tabular}} 
 &\multirow{2}{*}{\begin{tabular}{c}\textsf{IPW}\end{tabular}} 
 &\multirow{2}{*}{\begin{tabular}{c}\textsf{IPW}\end{tabular}} 
   & \multirow{2}{*}{\begin{tabular}{c}$\lambda_1=0.5$\end{tabular}}
   & \multirow{2}{*}{\begin{tabular}{c}$\lambda_1=0.6$\end{tabular}}
    &\multirow{2}{*}{\begin{tabular}{c}\textsf{AIPW}\end{tabular}} 
 &\multirow{2}{*}{\begin{tabular}{c}\textsf{AIPW}\end{tabular}} 
   & \multirow{2}{*}{\begin{tabular}{c}$\lambda_1=0.5$\end{tabular}}
   & \multirow{2}{*}{\begin{tabular}{c}$\lambda_1=0.6$\end{tabular}}
  \\  \\
  
  0.00 & 29.03 & 1.03 & 3.61 & 2.14  & 23.11 & 1.06 & 3.05 & 1.89 \\ 
  0.01 & 28.23 & 1.06 & 4.25 & 2.72  & 23.60 & 1.06 & 3.53 & 2.23 \\ 
  0.10 & 8.85 & 1.07 & 1.73 & 1.34  & 6.62 & 1.09 & 1.50 & 1.25 \\ 
  0.50 & 1.00 & 1.98 & 1.18 & 1.33  & 1.00 & 2.84 & 1.48 & 1.74 \\ 
  0.60 & 1.00 & 2.91 & 1.51 & 1.76  & 1.00 & 3.88 & 1.72 & 2.07 \\ 
  0.70 & 1.00 & 3.74 & 1.71 & 2.06  & 1.00 & 5.30 & 1.95 & 2.44 \\ 
  2.00 & 1.00 & 12.87 & 2.32 & 3.05  & 1.00 & 18.13 & 2.40 & 3.17 \\ 
  3.00 & 1.00 & 17.72 & 2.21 & 2.98  & 1.00 & 23.54 & 2.20 & 2.93 \\ 
  10.00 & 1.00 & 20.87 & 2.26 & 3.01 & 1.00 & 29.93 & 2.28 & 3.04 \\  
    \hline
    \hline
   \multicolumn{9}{c}{$\beta(X_1, X_2, X_3) = 2 - X_1 - X_2$} \\
    \hline

\multirow{2}{*}{\begin{tabular}{c}b\end{tabular}} 
 &\multirow{2}{*}{\begin{tabular}{c}\textsf{IPPW}\end{tabular}} 
 &\multirow{2}{*}{\begin{tabular}{c}\textsf{IPW}\end{tabular}} 
   & \multirow{2}{*}{\begin{tabular}{c}$\lambda_1=0.5$\end{tabular}}
   & \multirow{2}{*}{\begin{tabular}{c}$\lambda_1=0.6$\end{tabular}}
    &\multirow{2}{*}{\begin{tabular}{c}\textsf{AIPPW}\end{tabular}} 
 &\multirow{2}{*}{\begin{tabular}{c}\textsf{AIPW}\end{tabular}} 
   & \multirow{2}{*}{\begin{tabular}{c}$\lambda_1=0.5$\end{tabular}}
   & \multirow{2}{*}{\begin{tabular}{c}$\lambda_1=0.6$\end{tabular}}
  \\  \\  
  0.00 & 28.81 & 1.04 & 3.61 & 2.23 & 25.26 & 1.05 & 3.25 & 2.04 \\ 
   0.01 & 31.66 & 1.05 & 4.51 & 2.73 & 27.70 & 1.05 & 3.94 & 2.43 \\ 
   0.10 & 8.91 & 1.06 & 1.73 & 1.33 & 7.34 & 1.07 & 1.69 & 1.35 \\ 
   0.50 & 1.00 & 1.91 & 1.17 & 1.32 & 1.00 & 2.63 & 1.39 & 1.60 \\ 
   0.60 & 1.00 & 2.81 & 1.41 & 1.66 & 1.00 & 3.67 & 1.58 & 1.91 \\ 
   0.70 & 1.00 & 3.76 & 1.66 & 2.02 & 1.00 & 4.87 & 1.80 & 2.24 \\ 
   2.00 & 1.00 & 13.07 & 2.22 & 2.94 & 1.00 & 18.44 & 2.21 & 2.93 \\ 
   3.00 & 1.00 & 18.60 & 2.42 & 3.25 & 1.00 & 26.37 & 2.44 & 3.29 \\ 
   10.00 & 1.00 & 23.45 & 2.49 & 3.36 & 1.00 & 30.90 & 2.35 & 3.14 \\ \hline
\end{tabular}}
\end{table}

Table \ref{tbl: simulation results} summarizes the simulation results when $n_o = 10,000$. Results for $n_o = 50,000$ and $n_o = 100,000$ are similar and can be found in Appendix \ref{appx:more_simulation}. Under each data-generating process, we calculated the mean squared error of each estimator, including the oracle estimator $\hat\beta_{\orac}$, the RCT-based estimator $\hat\beta_c$, the observational-data-based estimator $\hat\beta_o$, and a collection of adaptive estimators $\hat\beta_{\lambda}$ with different choices of $\lambda_1$, and reported the MSE of each estimator as a multiple of that of the oracle estimator. The simulation results aligned well with theory and intuition. When $b$ and hence the level of unmeasured confounding is small, the RCT-based estimator $\hat\beta_c$ performs much poorly compared to the oracle estimator. For instance, when $\beta(X_1, X_2, X_3) = 2$ and $b = 0.01$, $\hat\beta_c$ has a mean squared error as large as $28$ times of that of $\hat\beta_\orac$ when $\hat\beta_c$ is constructed via inverse probability weighting and a mean squared error as large as $23$ times of that of $\hat\beta_\orac$ when constructed via augmented inverse probability weighting. At the other end of the spectrum, the observational-data-based estimator $\hat\beta_o$ performs poorly when there is a large degree of unmeasured confounding. On the other hand, the adaptive estimator $\hat\beta_{\lambda}$ has rather consistent performance across the spectrum of the level of unmeasured confounding, with a mean squared error within a small multiple of that of the oracle estimator. These patterns remain true when the treatment effect is homogeneous or heterogeneous, and when $\hat\beta_c$ and $\hat\beta_o$ are constructed in different ways.




\section{Real Data Example: the  RCT DUPLICATE Initiative}

\begin{table}[t]
    \centering
    \resizebox{\textwidth}{!}{
    \begin{tabular}{lcccccc}
       \multirow{4}{*}{\begin{tabular}{c}\textbf{Study Name}\end{tabular}}   & \multirow{4}{*}{\begin{tabular}{c}\textbf{Outcome}\end{tabular}}  & 
       \multirow{4}{*}{\begin{tabular}{c}\textbf{RCT} \\\textbf{Sample}\\\textbf{Size}\end{tabular}}  & 
       \multirow{4}{*}{\begin{tabular}{c}\textbf{RCT} \\\textbf{Estimate}\end{tabular}}  & 
       \multirow{4}{*}{\begin{tabular}{c}\textbf{OBS} \\\textbf{Sample}\\\textbf{Size}\end{tabular}}  & 
       \multirow{4}{*}{\begin{tabular}{c}\textbf{OBS} \\\textbf{Estimate}\end{tabular}}  & 
       \multirow{4}{*}{\begin{tabular}{c}\textbf{Anchored} \\\textbf{Thresholding} \\\textbf{Estimate}\end{tabular}} 
       \\ \\ \\ \\
       \hline
       \multirow{2}{*}{\begin{tabular}{c}\textsf{LEADER} \end{tabular}} & \multirow{2}{*}{\begin{tabular}{c}3P MACE \end{tabular}} & \multirow{2}{*}{\begin{tabular}{c}9,340 \end{tabular}} & \multirow{2}{*}{\begin{tabular}{c}-0.0183\\ (0.0072) \end{tabular}}  & \multirow{2}{*}{\begin{tabular}{c}168,692 \end{tabular}} & \multirow{2}{*}{\begin{tabular}{c}-0.0071 \\ (0.0007) \end{tabular}}  & \multirow{2}{*}{\begin{tabular}{c}-0.0075 \end{tabular}} \\ \\
       
        \multirow{2}{*}{\begin{tabular}{c}\textsf{DECLARE-TIMI 58} \end{tabular}} & \multirow{2}{*}{\begin{tabular}{c}HHF + death \end{tabular}} & \multirow{2}{*}{\begin{tabular}{c}17,160 \end{tabular}} & \multirow{2}{*}{\begin{tabular}{c}-0.0092\\ (0.0034) \end{tabular}}  & \multirow{2}{*}{\begin{tabular}{c}49,790 \end{tabular}} & \multirow{2}{*}{\begin{tabular}{c}-0.0050 \\(0.0010) \end{tabular}}  & \multirow{2}{*}{\begin{tabular}{c}-0.0053 \end{tabular}} \\ \\
        
        \multirow{2}{*}{\begin{tabular}{c} \textsf{EMPA-REG OUTCOME} \end{tabular}} & \multirow{2}{*}{\begin{tabular}{c}3P MACE \end{tabular}} & \multirow{2}{*}{\begin{tabular}{c}7,020 \end{tabular}} & \multirow{2}{*}{\begin{tabular}{c}-0.0163\\ (0.0081) \end{tabular}}  & \multirow{2}{*}{\begin{tabular}{c}103,750 \end{tabular}} & \multirow{2}{*}{\begin{tabular}{c}-0.0012\\ (0.0006) \end{tabular}}  & \multirow{2}{*}{\begin{tabular}{c}-0.0043 \end{tabular}} \\ \\
        
           \multirow{2}{*}{\begin{tabular}{c} \textsf{CANVAS} \end{tabular}} & \multirow{2}{*}{\begin{tabular}{c}3P MACE \end{tabular}} & \multirow{2}{*}{\begin{tabular}{c}10,142 \end{tabular}} & \multirow{2}{*}{\begin{tabular}{c}-0.0168 \\ (0.0062) \end{tabular}}  & \multirow{2}{*}{\begin{tabular}{c}152,198 \end{tabular}} & \multirow{2}{*}{\begin{tabular}{c}-0.0029\\ (0.0005) \end{tabular}}  & \multirow{2}{*}{\begin{tabular}{c}-0.0074 \end{tabular}} \\ \\
           
      \multirow{2}{*}{\begin{tabular}{c} \textsf{CARMELINA} \end{tabular}} & \multirow{2}{*}{\begin{tabular}{c}3P MACE \end{tabular}} & \multirow{2}{*}{\begin{tabular}{c}6,979 \end{tabular}} & \multirow{2}{*}{\begin{tabular}{c}0.0037\\ (0.0078) \end{tabular}}  & \multirow{2}{*}{\begin{tabular}{c}101,826 \end{tabular}} & \multirow{2}{*}{\begin{tabular}{c}-0.0056\\ (0.0011) \end{tabular}}  & \multirow{2}{*}{\begin{tabular}{c}-0.0054 \end{tabular}} \\ \\
      
         \multirow{2}{*}{\begin{tabular}{c} \textsf{TECOS} \end{tabular}} & \multirow{2}{*}{\begin{tabular}{c}3P MACE + angina \end{tabular}} & \multirow{2}{*}{\begin{tabular}{c}14,523 \end{tabular}} & \multirow{2}{*}{\begin{tabular}{c}-0.0015\\ (0.0053) \end{tabular}}  & \multirow{2}{*}{\begin{tabular}{c}349,478 \end{tabular}} & \multirow{2}{*}{\begin{tabular}{c}-0.0091 \\ (0.0007) \end{tabular}}  & \multirow{2}{*}{\begin{tabular}{c}-0.0089 \end{tabular}} \\ \\

       \multirow{2}{*}{\begin{tabular}{c} \textsf{SAVOR-TIMI 53} \end{tabular}} & \multirow{2}{*}{\begin{tabular}{c}3P MACE \end{tabular}} & \multirow{2}{*}{\begin{tabular}{c}16,492 \end{tabular}} & \multirow{2}{*}{\begin{tabular}{c}-0.0001\\ (0.0041) \end{tabular}}  & \multirow{2}{*}{\begin{tabular}{c}182,128 \end{tabular}} & \multirow{2}{*}{\begin{tabular}{c}-0.0080 \\ (0.0007) \end{tabular}}  & \multirow{2}{*}{\begin{tabular}{c}-0.0064 \end{tabular}} \\ \\
       
       \multirow{2}{*}{\begin{tabular}{c} \textsf{CAROLINA} \end{tabular}} & \multirow{2}{*}{\begin{tabular}{c}3P MACE \end{tabular}} & \multirow{2}{*}{\begin{tabular}{c}6,033 \end{tabular}} & \multirow{2}{*}{\begin{tabular}{c}-0.0025\\ (0.0083) \end{tabular}}  & \multirow{2}{*}{\begin{tabular}{c}48,262 \end{tabular}} & \multirow{2}{*}{\begin{tabular}{c}-0.0035\\ (0.0012) \end{tabular}}  & \multirow{2}{*}{\begin{tabular}{c}-0.0035 \end{tabular}} \\ \\
       
        \multirow{2}{*}{\begin{tabular}{c} \textsf{TRITON-TIMI 38} \end{tabular}} & \multirow{2}{*}{\begin{tabular}{c}3P MACE \end{tabular}} & \multirow{2}{*}{\begin{tabular}{c}13,608 \end{tabular}} & \multirow{2}{*}{\begin{tabular}{c}-0.0206\\ (0.0052) \end{tabular}}  & \multirow{2}{*}{\begin{tabular}{c}43,864 \end{tabular}} & \multirow{2}{*}{\begin{tabular}{c}-0.0110\\ (0.0018) \end{tabular}}  & \multirow{2}{*}{\begin{tabular}{c}-0.0129 \end{tabular}} \\ \\

         \multirow{2}{*}{\begin{tabular}{c} \textsf{PLATO} \end{tabular}} & \multirow{2}{*}{\begin{tabular}{c}3P MACE \end{tabular}} & \multirow{2}{*}{\begin{tabular}{c}18,624 \end{tabular}} & \multirow{2}{*}{\begin{tabular}{c}-0.0166\\ (0.0044) \end{tabular}}  & \multirow{2}{*}{\begin{tabular}{c}27,960 \end{tabular}} & \multirow{2}{*}{\begin{tabular}{c}-0.0149\\ (0.0027) \end{tabular}}  & \multirow{2}{*}{\begin{tabular}{c}-0.0154 \end{tabular}} \\ \\
       
    \end{tabular}}
    \caption{Results from 10 RCTs and their corresponding emulation studies. 3P MACE indicates 3-point major adverse cardiovascular events (myocardial infarction, stroke, or cardiovascular death). }
    \label{tab: real data summary table}
\end{table}
To assess whether non-randomized, observational claims data could provide accurate causal estimates of medical products, \cite{franklin2020nonrandomized, Franklin:2021aa} launched the RCT DUPLICATE initiative (Randomized, Controlled Trials Duplicated Using Prospective Longitudinal Insurance Claims: Applying Techniques of Epidemiology) where findings from a selected cohort of RCTs are compared to findings from observational health care claims data investigating the same clinical questions. Specifically, \cite{Franklin:2021aa} emulated the study design of $8$ cardiovascular outcome trials of antidiabetic medications and $2$ trials of antiplatelets using observational claims data by selecting the same primary outcome, treatment strategy, and applying the same inclusion and exclusion criteria so that the RCT cohort and observational cohort are similar in observed covariates. To estimate the effect of intervention from observational claims data, \cite{Franklin:2021aa} used $1$-to-$1$ propensity-score matching (\citealp{rosenbaum1983central,rosenbaum1985constructing}) and embedded observational data into an approximate randomized controlled trial. Just like any covariate adjustment methodology, statistical matching can only remove the overt bias but not the hidden bias due to unmeasured confounding; therefore, the effect estimate based on the observational data could be potentially biased. Using published summary data in \cite{Franklin:2021aa}, we calculated the effect of each intervention on event rate (defined as the rate that a pre-specified composite outcome, e.g., 3-point major adverse cardiovascular events, occurs) in each of the $10$ RCTs and their corresponding matched observational studies. Table \ref{tab: real data summary table} summarizes the results from 10 RCTs and their corresponding observational study emulations. Overall, we observe that the observational-data-based estimator and the RCT-based estimator align reasonably well. We then applied the adaptive procedure to pool together the RCT-based and observational-data-based estimates. We set the hyper-parameter $\lambda_1 = 0.5$ as guided by our simulation studies. The anchored thresholding estimates agreed with the RCT-only estimates in sign in all but the CARMELINA trial, and fell within the $95\%$ confidence interval of the RCT-only estimates in all $10$ trials. As demonstrated by our extensive simulations, the anchored thresholding estimates are expected to have near optimal mean square errors regardless of the magnitude of hidden bias in observational study emulations.

\section{Discussion}
In this paper, we characterize the potential efficiency gain from integrating observational data into the RCT-based analysis by establishing the minimax rates for both estimation and confidence interval construction. For estimation, we propose a fully adaptive estimator that nearly achieves the minimax rate; for inference, we show that adaptation is in general not possible without additional knowledge on the magnitude of the bias.


The anchored thresholding estimator proposed in Section \ref{subsec: adaptive_procedure} is a general strategy for aggregating two estimators, where one is unbiased and the other is not. This strategy may be extended to other application scenarios, such as estimating individualized treatment rules \citep{wu2021transfer}.

\small{
\setlength{\bibsep}{0.2pt plus 0.3ex}
\bibliographystyle{abbrvnat}
\bibliography{references}

\begin{thebibliography}{46}
\providecommand{\natexlab}[1]{#1}
\providecommand{\url}[1]{\texttt{#1}}
\expandafter\ifx\csname urlstyle\endcsname\relax
  \providecommand{\doi}[1]{doi: #1}\else
  \providecommand{\doi}{doi: \begingroup \urlstyle{rm}\Url}\fi

\bibitem[Aronow and Middleton(2013)]{aronow2013class}
P.~M. Aronow and J.~A. Middleton.
\newblock A class of unbiased estimators of the average treatment effect in
  randomized experiments.
\newblock \emph{Journal of Causal Inference}, 1\penalty0 (1):\penalty0
  135--154, 2013.

\bibitem[Bang and Robins(2005)]{bang2005doubly}
H.~Bang and J.~M. Robins.
\newblock Doubly robust estimation in missing data and causal inference models.
\newblock \emph{Biometrics}, 61\penalty0 (4):\penalty0 962--973, 2005.

\bibitem[Bareinboim and Pearl(2016)]{bareinboim2016causal}
E.~Bareinboim and J.~Pearl.
\newblock Causal inference and the data-fusion problem.
\newblock \emph{Proceedings of the National Academy of Sciences}, 113\penalty0
  (27):\penalty0 7345--7352, 2016.

\bibitem[Boucheron et~al.(2013)Boucheron, Lugosi, and
  Massart]{boucheron2013concentration}
S.~Boucheron, G.~Lugosi, and P.~Massart.
\newblock \emph{Concentration inequalities: A nonasymptotic theory of
  independence}.
\newblock Oxford university press, 2013.

\bibitem[Buchanan et~al.(2018)Buchanan, Hudgens, Cole, Mollan, Sax, Daar,
  Adimora, Eron, and Mugavero]{Buchanan:2018aa}
A.~L. Buchanan, M.~G. Hudgens, S.~R. Cole, K.~R. Mollan, P.~E. Sax, E.~S. Daar,
  A.~A. Adimora, J.~J. Eron, and M.~J. Mugavero.
\newblock Generalizing evidence from randomized trials using inverse
  probability of sampling weights.
\newblock \emph{Journal of the Royal Statistical Society. Series A, (Statistics
  in Society)}, 181\penalty0 (4):\penalty0 1193--1209, 10 2018.

\bibitem[Cai and Guo(2017)]{cai2017confidence}
T.~T. Cai and Z.~Guo.
\newblock Confidence intervals for high-dimensional linear regression: Minimax
  rates and adaptivity.
\newblock \emph{The Annals of statistics}, 45\penalty0 (2):\penalty0 615--646,
  2017.

\bibitem[Cai and Low(2004)]{cai2004adaptation}
T.~T. Cai and M.~G. Low.
\newblock An adaptation theory for nonparametric confidence intervals.
\newblock \emph{The Annals of statistics}, 32\penalty0 (5):\penalty0
  1805--1840, 2004.

\bibitem[Candes and Tao(2007)]{candes2007dantzig}
E.~Candes and T.~Tao.
\newblock The dantzig selector: Statistical estimation when p is much larger
  than n.
\newblock \emph{The annals of Statistics}, 35\penalty0 (6):\penalty0
  2313--2351, 2007.

\bibitem[Cand{\`e}s and Recht(2009)]{candes2009exact}
E.~J. Cand{\`e}s and B.~Recht.
\newblock Exact matrix completion via convex optimization.
\newblock \emph{Foundations of Computational mathematics}, 9\penalty0
  (6):\penalty0 717--772, 2009.

\bibitem[Cand{\`e}s and Tao(2010)]{candes2010power}
E.~J. Cand{\`e}s and T.~Tao.
\newblock The power of convex relaxation: Near-optimal matrix completion.
\newblock \emph{IEEE Transactions on Information Theory}, 56\penalty0
  (5):\penalty0 2053--2080, 2010.

\bibitem[Chauss{\'e}(2010)]{gmmpackage}
P.~Chauss{\'e}.
\newblock Computing generalized method of moments and generalized empirical
  likelihood with {R}.
\newblock \emph{Journal of Statistical Software}, 34\penalty0 (11):\penalty0
  1--35, 2010.
\newblock URL \url{https://www.jstatsoft.org/v34/i11/}.

\bibitem[Chen et~al.(2021)Chen, Zheng, Long, and Su]{chen2021theorem}
S.~Chen, Q.~Zheng, Q.~Long, and W.~J. Su.
\newblock A theorem of the alternative for personalized federated learning.
\newblock \emph{arXiv preprint arXiv:2103.01901}, 2021.

\bibitem[Cole and Stuart(2010)]{cole2010generalizing}
S.~R. Cole and E.~A. Stuart.
\newblock Generalizing evidence from randomized clinical trials to target
  populations: the actg 320 trial.
\newblock \emph{American journal of epidemiology}, 172\penalty0 (1):\penalty0
  107--115, 2010.

\bibitem[Colnet et~al.(2020)Colnet, Mayer, Chen, Dieng, Li, Varoquaux, Vert,
  Josse, and Yang]{colnet2020causal}
B.~Colnet, I.~Mayer, G.~Chen, A.~Dieng, R.~Li, G.~Varoquaux, J.-P. Vert,
  J.~Josse, and S.~Yang.
\newblock Causal inference methods for combining randomized trials and
  observational studies: a review, 2020.

\bibitem[Dahabreh et~al.(2019)Dahabreh, Robertson, Tchetgen, Stuart, and
  Hern{\'a}n]{Dahabreh:2019aa}
I.~J. Dahabreh, S.~E. Robertson, E.~J. Tchetgen, E.~A. Stuart, and M.~A.
  Hern{\'a}n.
\newblock Generalizing causal inferences from individuals in randomized trials
  to all trial-eligible individuals.
\newblock \emph{Biometrics}, 75\penalty0 (2):\penalty0 685--694, 2021/05/23
  2019.

\bibitem[Deaton and Cartwright(2018)]{deaton2018understanding}
A.~Deaton and N.~Cartwright.
\newblock Understanding and misunderstanding randomized controlled trials.
\newblock \emph{Social Science \& Medicine}, 210:\penalty0 2--21, 2018.

\bibitem[Deaton(2009)]{deaton2009instruments}
A.~S. Deaton.
\newblock Instruments of development: Randomization in the tropics, and the
  search for the elusive keys to economic development.
\newblock Technical report, National Bureau of Economic Research, 2009.

\bibitem[Degtiar and Rose(2021)]{degtiar2021review}
I.~Degtiar and S.~Rose.
\newblock A review of generalizability and transportability, 2021.

\bibitem[Feller(2008)]{feller2008introduction}
W.~Feller.
\newblock \emph{An introduction to probability theory and its applications, vol
  2}.
\newblock John Wiley \& Sons, 2008.

\bibitem[Franklin et~al.(2020)Franklin, Pawar, Martin, Glynn, Levenson, Temple,
  and Schneeweiss]{franklin2020nonrandomized}
J.~M. Franklin, A.~Pawar, D.~Martin, R.~J. Glynn, M.~Levenson, R.~Temple, and
  S.~Schneeweiss.
\newblock Nonrandomized real-world evidence to support regulatory decision
  making: process for a randomized trial replication project.
\newblock \emph{Clinical Pharmacology \& Therapeutics}, 107\penalty0
  (4):\penalty0 817--826, 2020.

\bibitem[Franklin et~al.(2021{\natexlab{a}})Franklin, Patorno, Desai, Glynn,
  Martin, Quinto, Pawar, Bessette, Lee, Garry, Gautam, and
  Schneeweiss]{Franklin:2021aa}
J.~M. Franklin, E.~Patorno, R.~J. Desai, R.~J. Glynn, D.~Martin, K.~Quinto,
  A.~Pawar, L.~G. Bessette, H.~Lee, E.~M. Garry, N.~Gautam, and S.~Schneeweiss.
\newblock Emulating randomized clinical trials with nonrandomized real-world
  evidence studies.
\newblock \emph{Circulation}, 143\penalty0 (10):\penalty0 1002--1013,
  2021{\natexlab{a}}.

\bibitem[Franklin et~al.(2021{\natexlab{b}})Franklin, Patorno, Desai, Glynn,
  Martin, Quinto, Pawar, Bessette, Lee, Garry, et~al.]{franklin2021emulating}
J.~M. Franklin, E.~Patorno, R.~J. Desai, R.~J. Glynn, D.~Martin, K.~Quinto,
  A.~Pawar, L.~G. Bessette, H.~Lee, E.~M. Garry, et~al.
\newblock Emulating randomized clinical trials with nonrandomized real-world
  evidence studies: first results from the rct duplicate initiative.
\newblock \emph{Circulation}, 143\penalty0 (10):\penalty0 1002--1013,
  2021{\natexlab{b}}.

\bibitem[Gagnon-Bartsch et~al.(2021)Gagnon-Bartsch, Sales, Wu, Botelho,
  Erickson, Miratrix, and Heffernan]{gagnon2021precise}
J.~A. Gagnon-Bartsch, A.~C. Sales, E.~Wu, A.~F. Botelho, J.~A. Erickson, L.~W.
  Miratrix, and N.~T. Heffernan.
\newblock Precise unbiased estimation in randomized experiments using auxiliary
  observational data.
\newblock \emph{arXiv preprint arXiv:2105.03529}, 2021.

\bibitem[Hartman et~al.(2015)Hartman, Grieve, Ramsahai, and
  Sekhon]{hartman2015sample}
E.~Hartman, R.~Grieve, R.~Ramsahai, and J.~S. Sekhon.
\newblock From sample average treatment effect to population average treatment
  effect on the treated: combining experimental with observational studies to
  estimate population treatment effects.
\newblock \emph{Journal of the Royal Statistical Society. Series A (Statistics
  in Society)}, pages 757--778, 2015.

\bibitem[Horvitz and Thompson(1952)]{horvitz1952generalization}
D.~G. Horvitz and D.~J. Thompson.
\newblock A generalization of sampling without replacement from a finite
  universe.
\newblock \emph{Journal of the American statistical Association}, 47\penalty0
  (260):\penalty0 663--685, 1952.

\bibitem[Imbens(2003)]{Imbens2003}
G.~W. Imbens.
\newblock {Sensitivity to exogeneity assumptions in program evaluation}.
\newblock \emph{American Economic Review}, 93:\penalty0 126--132, 2003.

\bibitem[{JAMA Network}(2021)]{JAMA_instructions}
{JAMA Network}.
\newblock Instruction for authors.
\newblock
  \url{https://jamanetwork.com/journals/jama/pages/instructions-for-authors},
  2021.

\bibitem[Lu et~al.(2019)Lu, Scharfstein, Brooks, Quach, and
  Kennedy]{lu2019causal}
Y.~Lu, D.~O. Scharfstein, M.~M. Brooks, K.~Quach, and E.~H. Kennedy.
\newblock Causal inference for comprehensive cohort studies.
\newblock \emph{arXiv preprint arXiv:1910.03531}, 2019.

\bibitem[Newey and McFadden(1994)]{newey1994}
W.~K. Newey and D.~McFadden.
\newblock Chapter 36 large sample estimation and hypothesis testing.
\newblock 4:\penalty0 2111--2245, 1994.
\newblock ISSN 1573-4412.
\newblock \doi{https://doi.org/10.1016/S1573-4412(05)80005-4}.
\newblock URL
  \url{https://www.sciencedirect.com/science/article/pii/S1573441205800054}.

\bibitem[Pearl and Bareinboim(2014)]{pearl2014external}
J.~Pearl and E.~Bareinboim.
\newblock External validity: From do-calculus to transportability across
  populations.
\newblock \emph{Statistical Science}, 29\penalty0 (4):\penalty0 579--595, 2014.

\bibitem[Robins et~al.(2007)Robins, Sued, Lei-Gomez, and
  Rotnitzky]{robins2007comment}
J.~Robins, M.~Sued, Q.~Lei-Gomez, and A.~Rotnitzky.
\newblock Comment: Performance of double-robust estimators when" inverse
  probability" weights are highly variable.
\newblock \emph{Statistical Science}, 22\penalty0 (4):\penalty0 544--559, 2007.

\bibitem[Robins et~al.(1994)Robins, Rotnitzky, and Zhao]{robins1994estimation}
J.~M. Robins, A.~Rotnitzky, and L.~P. Zhao.
\newblock Estimation of regression coefficients when some regressors are not
  always observed.
\newblock \emph{Journal of the American statistical Association}, 89\penalty0
  (427):\penalty0 846--866, 1994.

\bibitem[Rosenbaum and Rubin(1983{\natexlab{a}})]{rosenbaum1983central}
P.~R. Rosenbaum and D.~B. Rubin.
\newblock The central role of the propensity score in observational studies for
  causal effects.
\newblock \emph{Biometrika}, 70\penalty0 (1):\penalty0 41--55,
  1983{\natexlab{a}}.

\bibitem[Rosenbaum and Rubin(1983{\natexlab{b}})]{rosenbaum1983sens}
P.~R. Rosenbaum and D.~B. Rubin.
\newblock {Assessing sensitivity to an unobserved binary covariate in an
  observational study with binary outcome}.
\newblock \emph{Journal of Royal Statistical Society, Series B}, 45:\penalty0
  212--218, 1983{\natexlab{b}}.

\bibitem[Rosenbaum and Rubin(1985)]{rosenbaum1985constructing}
P.~R. Rosenbaum and D.~B. Rubin.
\newblock Constructing a control group using multivariate matched sampling
  methods that incorporate the propensity score.
\newblock \emph{The American Statistician}, 39\penalty0 (1):\penalty0 33--38,
  1985.

\bibitem[Rothwell(2005)]{rothwell2005external}
P.~M. Rothwell.
\newblock External validity of randomised controlled trials:“to whom do the
  results of this trial apply?”.
\newblock \emph{The Lancet}, 365\penalty0 (9453):\penalty0 82--93, 2005.

\bibitem[Rubin(1980)]{Rubin1980}
D.~B. Rubin.
\newblock {Randomization analysis of experimental data: the Fisher
  randomization test comment}.
\newblock \emph{Journal of the American Statistical Association}, 75:\penalty0
  591--593, 1980.

\bibitem[Stuart et~al.(2011)Stuart, Cole, Bradshaw, and Leaf]{Stuart:2011aa}
E.~A. Stuart, S.~R. Cole, C.~P. Bradshaw, and P.~J. Leaf.
\newblock The use of propensity scores to assess the generalizability of
  results from randomized trials.
\newblock \emph{Journal of the Royal Statistical Society. Series A, (Statistics
  in Society)}, 174\penalty0 (2):\penalty0 369--386, 04 2011.

\bibitem[Tibshirani(1996)]{tibshirani1996regression}
R.~Tibshirani.
\newblock Regression shrinkage and selection via the lasso.
\newblock \emph{Journal of the Royal Statistical Society: Series B
  (Methodological)}, 58\penalty0 (1):\penalty0 267--288, 1996.

\bibitem[Tipton(2013)]{tipton2013improving}
E.~Tipton.
\newblock Improving generalizations from experiments using propensity score
  subclassification: Assumptions, properties, and contexts.
\newblock \emph{Journal of Educational and Behavioral Statistics}, 38\penalty0
  (3):\penalty0 239--266, 2013.

\bibitem[Tsiatis(2007)]{tsiatis2007semiparametric}
A.~Tsiatis.
\newblock \emph{Semiparametric theory and missing data}.
\newblock Springer Science \& Business Media, 2007.

\bibitem[Tsiatis et~al.(2008)Tsiatis, Davidian, Zhang, and
  Lu]{tsiatis2008covariate}
A.~A. Tsiatis, M.~Davidian, M.~Zhang, and X.~Lu.
\newblock Covariate adjustment for two-sample treatment comparisons in
  randomized clinical trials: a principled yet flexible approach.
\newblock \emph{Statistics in medicine}, 27\penalty0 (23):\penalty0 4658--4677,
  2008.

\bibitem[Wu and Yang(2021)]{wu2021transfer}
L.~Wu and S.~Yang.
\newblock Transfer learning of individualized treatment rules from experimental
  to real-world data.
\newblock \emph{arXiv preprint arXiv:2108.08415}, 2021.

\bibitem[Yang et~al.(2020{\natexlab{a}})Yang, Zeng, and Wang]{yang2020elastic}
S.~Yang, D.~Zeng, and X.~Wang.
\newblock Elastic integrative analysis of randomized trial and real-world data
  for treatment heterogeneity estimation.
\newblock \emph{arXiv preprint arXiv:2005.10579}, 2020{\natexlab{a}}.

\bibitem[Yang et~al.(2020{\natexlab{b}})Yang, Zeng, and Wang]{yang2020improved}
S.~Yang, D.~Zeng, and X.~Wang.
\newblock Improved inference for heterogeneous treatment effects using
  real-world data subject to hidden confounding, 2020{\natexlab{b}}.

\bibitem[Ye et~al.(2020)Ye, Shao, Yi, and Zhao]{ye2020principles}
T.~Ye, J.~Shao, Y.~Yi, and Q.~Zhao.
\newblock Toward better practice of covariate adjustment in analyzing
  randomized clinical trials.
\newblock \emph{arXiv preprint arXiv:2009.11828}, 2020.

\end{thebibliography}
}

\newpage
\appendixtitleon
\appendixtitletocon
\addtocontents{toc}{\protect\setcounter{tocdepth}{2}}
\begin{appendices}

\section{Additional Details}

\subsection{Bias Analysis of the Observational-Data-Only Estimators}\label{appx:internal_bias}
 To understand how the internal validity bias arises from unobserved confounding, we consider an unmeasured confounder $U$ such that the treatment assignment becomes ignorable once conditional on $(\mathbf{X}, U)$ \citep{rosenbaum1983sens,rosenbaum1983central, Imbens2003}, i.e., $ A  \indep ( Y(1), Y(0)) \mid \mathbf{X}, U, S=o$ but $ A \nindep ( Y(1), Y(0)) \mid \mathbf{X}, S=o$. Hence,
 \begin{align*}
\beta_o^\star&= \bbE( Y {(1)}  \mid S=o ) -  \bbE( Y {(0)}  \mid S=o ) \\
	&= \bbE\left[ \frac{A Y }{\pi_o(\bfX, U)} \mid S= o \right]  - \bbE\left[ \frac{(1-A) Y}{1-\pi_o(\bfX, U) } \mid S= o \right] \\
	&=   \bbE\left[ \frac{AY}{\pi_o(\bfX)  } \left\{ \frac{1-\pi_o(\bfX)}{\gamma (\bfX, U)}  + \pi_o(\bfX)  \right\}  \mid S= o \right]  \\
	&\qquad - \bbE\left[ \frac{(1-A) Y}{1-\pi_o(\bfX)} \left\{ \pi_o(\bfX)\gamma(\bfX, U)  + 1-\pi_o(\bfX) \right\}  \mid S= o \right],     
\end{align*}
where $\gamma(\mathbf{X}, U) = \text{OR}(\pi_{o}(\mathbf{X}, U), \pi_{o}(\mathbf{X}))$ is the odds ratio of the true propensity score $\pi_{o}(\mathbf{X}, U) = P(A= 1 \mid \mathbf{X}, U, S = o)$ that takes into account the unmeasured confounder $U$ and the observed propensity score $\pi_{o}(\mathbf{X}) = \bbP(A = 1 \mid \mathbf{X}, S = o)$.
Given that
\begin{align*}
    \beta_o= \bbE\left[ \frac{A Y }{\pi_o(\bfX)} \mid S= o \right]  - \bbE\left[ \frac{(1-A) Y}{1-\pi_o(\bfX) } \mid S= o \right],
\end{align*}
one can readily check that the difference $\Delta = \beta^\star - \beta_\obs =  \beta_o^\star - \beta_o$ takes the following form:
\begin{align*}
  \Delta & =  \mathbb{E}\left[ \frac{(1-A) Y\pi_{o}(\mathbf{X})}{1-\pi_{o}(\mathbf{X})} \left\{  1-\gamma(\mathbf{X}, U)  \right\}  \mid S= o \right]\\
  & \qquad - \mathbb{E}\left[ \frac{AY(1-\pi_{o}(\mathbf{X}))}{\pi_{o}(\mathbf{X})  } \left\{   1- \frac{1}{\gamma (\mathbf{X}, U)}  \right\}  \mid S= o \right].
\end{align*}
In this way, $\Delta$ measures the internal validity bias from the observational data. 
There are other ways to parametrize and represent the bias $\Delta$. Our method and theory developed in Sections \ref{sec:estimation} and \ref{sec:inference} do \emph{not} depend on the particular parametrization and allows for arbitrary bias from the observational-data-only estimator.

\subsection{Stabilized Estimators}\label{appx:stablization}
It is well know that IPW-type estimators like $\cippw, \caippw, \ooipw$, and $\ooaipw$ can be unstable when the (estimated) weights are close to zero. To alleviate this issue, one can apply the stabilization method by normalizing the weights to sum to one \citep{robins2007comment}. The stabilized estimators are asymptotically equivalent to their unstabilized counterparts but they have better finite-sample performance.
For example, the stabilized version of $\caippw$ is 
\begin{align*}
        &\hat\beta_{c, \textsf{sAIPPW}} \\
        &= \bigg[  \frac{1}{n_c} \sum_{i\in \mathcal{C}} \frac{A_i(1-e_{{c}}(\mathbf{X}_i; \hat \bsxi )) }{\pi_c(\mathbf{X}_i)e_{{c}}(\mathbf{X}_i; \hat \bsxi ) }\bigg]^{-1} \frac{1}{n_c} \sum_{i\in \mathcal{C}} \frac{A_i(1-e_{{c}}(\mathbf{X}_i; \hat \bsxi )) (Y_i -\mu_{c1} (\mathbf{X}_i; \hat \bsalpha_{c1} ) )  }{\pi_c(\mathbf{X}_i)e_{{c}}(\mathbf{X}_i; \hat \bsxi ) } \\
        &  - \bigg[  \frac{1}{n_c} \sum_{i\in \mathcal{C}} \frac{(1-A_i)(1-e_{{c}}(\mathbf{X}_i; \hat \bsxi ))}{(1- \pi_c(\mathbf{X}_i))e_{{c}}(\mathbf{X}_i; \hat \bsxi )} \bigg]^{-1}\frac{1}{n_c} \sum_{i\in \mathcal{C}} \frac{(1-A_i)(1-e_{{c}}(\mathbf{X}_i; \hat \bsxi ))(Y_i -\mu_{c0} (\mathbf{X}_i; \hat \bsalpha_{c0} ) ) }{(1- \pi_c(\mathbf{X}_i))e_{{c}}(\mathbf{X}_i; \hat \bsxi )}   \\
& + \frac{1}{n_c} \sum_{i\in \mathcal{O}_1} \mu_{c1} (\mathbf{X}_i; \hat \bsalpha_{c1} )   - \mu_{c0} (\mathbf{X}_i; \hat \bsalpha_{c0} ). 
\end{align*}
and the stabilized AIPW estimator version of $\ooaipw$ is given by 
\begin{align*}
        \hat\beta_{o_2, \textsf{sAIPW}} &=
        \bigg[\frac{1}{n_{o_2}} \sum_{i\in \mathcal{O}_2}\frac{A_i  }{\pi_{o}(\mathbf{X}_i; \hat\bseta )} \bigg]^{-1}\frac{1}{n_{o_2}} \sum_{i\in \mathcal{O}_2} \frac{A_i (Y_i -\mu_{o1} (\mathbf{X}_i; \hat \bsalpha_{o1} ) )  }{\pi_{o}(\mathbf{X}_i; \hat\bseta )}  \\
        & \qquad - \bigg[\frac{1}{n_{o_2}} \sum_{i\in \mathcal{O}_2}\frac{1-A_i  }{1-\pi_{o}(\mathbf{X}_i; \hat\bseta )} \bigg]^{-1}\frac{1}{n_{o_2}} \sum_{i\in \mathcal{O}_2}\frac{(1-A_i)(Y_i -\mu_{o0} (\mathbf{X}_i; \hat \bsalpha_{o0} ) ) }{1- \pi_{o}(\mathbf{X}_i;\hat\bseta)}\\
        &    \qquad   +\frac{1}{n_{o_2}} \sum_{i\in \mathcal{O}_2} \mu_{o1} (\mathbf{X}_i; \hat \bsalpha_{o1} )   - \mu_{o0} (\mathbf{X}_i; \hat \bsalpha_{o0} ). 
\end{align*}

\subsection{Influence Function Representation}\label{appx:influence_func}
The estimators $\cippw, \caippw, \ooipw$, and $\ooaipw$ discussed in Section \ref{sec: setup} of the main article are generally referred to as the two-step estimators, which depend on some ``first-step''  estimators of nuisance parameters. A general device to analyze the two-step estimators is to stack the moment equations from the first step and the second step and apply the generalized method of moments (GMM) framework \citep[Chapter 6.1]{newey1994}. 

Take $\ooaipw$, the AIPW estimator based on the observational study, as an example.  Suppose that the first component of $\bfX$ is 1, representing the intercept. Let $O:= (Y, \bfX, A)$ denote the observed data vector and let $\bsphi = (\bsalpha_{o0}, \bsalpha_{o1},  \bseta) $ collect all nuisance parameters.  Then $\ooaipw$ can be written as the solution to a moment equation $n_{o_2}^{-1} \sum_{i\in \mathcal{O}_2}  g
(O_i, \beta; \hat \bsphi) = 0 $, where 
\begin{align*}
    g(O, \beta; \bsphi) &= \frac{A (Y- \mu_{o1} (\bfX; \bsalpha_{o1} )) }{\pi_o (\bfX; \bseta)} - \frac{(1-A) (Y-\mu_{o0} (\bfX; \bsalpha_{o0} )) }{1- \pi_o (\bfX; \bseta)} \\
    &\qquad + \mu_{o1} (\bfX; \bsalpha_{o1} )  - \mu_{o0} (\bfX; \bsalpha_{o0} ) - \beta. 
\end{align*}
The nuisance parameter estimator $\hat\bsphi$ can also be written as the solution to a set of moment equations $n_{o_2}^{-1} \sum_{i\in \mathcal{O}_2} m(O_i,  \bsphi) = 0 $, with $m(O, \bsphi)  = ( m(O, \bsalpha_{o0})^\top,  m(O, \bsalpha_{o1})^\top,  m(O, \bseta)^\top )^\top$, where 
\begin{align*}
       m(O, \bsalpha_{o0})& = (1- A) \bfX (Y-  \bsalpha_{o0}^\top \bfX), \\
       m(O, \bsalpha_{o1})& = A \bfX (Y-  \bsalpha_{o1}^\top \bfX), \\
       m(O, \bseta)& = \bfX (A-  \textsf{expit}(\bseta^\top \bfX)). 
\end{align*}
The definition of the moment equations $g $ and $m$, and the nuisance parameter $\bsphi$ should depend on the focal estimator, but we omit the subscripts for notational simplicity. 

 Let $\bar \bsphi$ be the solution to $\bbE[m(O, \bsphi)\mid S=o]= 0$ and let 
\begin{align*}
   &   G_\bsphi = \bbE [ \nabla_\bsphi g(O, \beta_o;\bar\bsphi)  \mid S=o ], \qquad M =  \bbE [ \nabla_\bsphi m(O, \bar\bsphi)  \mid S=o ],\\
    & g(O) = g(O, \beta_o;\bar\bsphi),  \qquad m(O) = m(O, \bar\bsphi). 
\end{align*}
By Theorem 6.1 of \cite{newey1994}, when either $\pi_o (\bfX; \bseta)$ or $\{\mu_{o0} (\bfX; \bsalpha_{o0} ), \mu_{o1} (\bfX; \bsalpha_{o1} )\}$ is correct,  we have that $\ooaipw$ is asymptotically linear, i.e., 
\begin{align*}
    \sqrt{n_{o_2}} (\ooaipw - \beta_o) =   \sum_{i\in \mathcal{O}_2} \{ g(O_i) - G_\bsphi M^{-1} m(O_i)  \} / \sqrt{n_{o_2}} +o_p(1). 
\end{align*}
Moreover, $\ooaipw$ is consistent and asymptotically normal, satisfying 
\begin{align*}
    \sqrt{n_{o_2}} (\ooaipw - \beta_o) \xrightarrow{d} N\big(0,  \bbE[ \{ g(O) - G_\bsphi M^{-1} m(O) \}^2\mid S=o ] \big). 
\end{align*}
This also gives a consistent variance estimator of $\ooaipw$ by replacing the population-level quantities in its asymptotic variance with their sample analogues. In practice, there is an easy way to obtain the variance estimator: first stack $g(O, \beta;\bsphi)$ and $m(O, \phi)$ together to form $\tilde g(O,\beta;\bsphi)= (g(O, \beta;\bsphi)^\top,m(O, \phi)^\top)^\top $, then solve for estimators of $(\beta, \bsphi^\top)^\top$ simultaneously using GMM (e.g., using the \textsf{gmm} package in \textsf{R} \citep{gmmpackage}). The variance estimator of $\ooaipw$ can be obtained immediately from the corresponding diagonal component of the outputted variance-covariance matrix.
The other estimators $\cippw, \caippw $ and $ \ooipw$ can be analyzed in the same way; the details are omitted.

\newpage
\section{Technical Proofs}

\subsection{Proof of Theorem \ref{thm:oracle_estimator}}\label{prf:thm:oracle_estimator}
The proof relies on the following version of Berry-Esseen theorem.
\begin{theorem}[Berry-Esseen theorem]
\label{thm:berry_esseen}
Let $X_1, \hdots, X_n$ be independent random variables with 
$$
    \bbE[X_i] = 0, \qquad \bbE[X_i^2] = \sigma_i^2, \qquad \bbE[|X_i|^3]  = \rho_i, \qquad \forall i\in[n].
$$
Then 
$$
    \sup_{t\in\bbR} \bigg|\bbP\bigg( \frac{\sum_{i\in[n]} X_i}{\sqrt{\sum_{i\in[n]} \sigma_i^2}} \leq t\bigg) - \Phi(t)\bigg| \leq \frac{6\sum_{i\in[n]} \rho_i}{(\sum_{i\in[n]} \sigma_i^2)^{3/2}},
$$
where $\Phi(t)$ is the CDF of $N(0, 1)$.
\end{theorem}
\begin{proof}
    See, e.g., Theorem 2 in Chapter 16 of \cite{feller2008introduction}.
\end{proof}

We will in fact prove the following stronger result:
\begin{equation}
    \label{eq:amse_convergence_convex_comb}
    \sup_{{\tilde{\weight}}\in[0, 1]} \bigg|\frac{\bbE[|\hat\beta_{\tilde{\weight}} - \beta^\star|^2]}{\amse(\hat\beta_{\tilde{\weight}})} - 1\bigg| \to 0 \qquad \textnormal{as }  n_\rct \land n_{\obs} \to \infty,
\end{equation}  
where 
\begin{equation*}
    \amse(\hat\beta_{\tilde{\weight}}) = 
    \frac{(1-{\tilde{\weight}})^2 \sigma_\rct^2}{n_\rct} + \frac{{\tilde{\weight}}^2 \sigma_\obs^2}{n_{\obs}} + {\tilde{\weight}}^2\Delta^2.
\end{equation*}
The formula for $\amse(\hat\beta_\rct)$ and $\amse(\hat\beta_w)$ follows by setting $\tilde\weight = 0$ and $\tilde\weight = \weight$, respectively. As an immediate result, we have
$$
    \amse(\hat\beta_\orac) = 
    \begin{cases}
        \frac{\sigma_\rct^2}{n_\rct} & \textnormal{ if }  \overline\Delta \geq \sigma_c/\sqrt{n_\rct}\\
        \frac{\sigma_\rct^2 \sigma_\obs^2}{n_\rct \sigma_\obs^2 + n_{\obs} \sigma_\rct^2} + \weight^2 \Delta^2 & \textnormal{ otherwise},
    \end{cases}.
$$
If $\overline\Delta \geq \sigma_\rct/\sqrt{n_\rct}$, we have
$$
    \amse(\hat\beta_\orac) = \frac{\sigma_\rct^2}{n_\rct}  = \frac{\sigma_\rct^2}{n_\rct}\land \overline{\Delta}^2 \leq 
    \frac{\sigma_\rct^2 \sigma_\obs^2}{n_\rct \sigma_\obs^2 + n_{\obs} \sigma_\rct^2} +  \frac{\sigma_\rct^2}{n_\rct}\land {\overline{\Delta}^2}.
$$
If $\overline\Delta < \sigma_\rct/\sqrt{n_\rct}$, we have
$$
   \amse(\hat\beta_\orac) \leq  \frac{\sigma_\rct^2 \sigma_\obs^2}{n_\rct \sigma_\obs^2 + n_{\obs} \sigma_\rct^2} + \weight^2 \Delta^2 \leq \frac{\sigma_\rct^2 \sigma_\obs^2}{n_\rct \sigma_\obs^2 + n_{\obs} \sigma_\rct^2} + \overline\Delta^2 = \frac{\sigma_\rct^2 \sigma_\obs^2}{n_\rct \sigma_\obs^2 + n_{\obs} \sigma_\rct^2} +  \frac{\sigma_\rct^2}{n_\rct}\land \overline{\Delta}^2.
$$
The formula for $\amse(\hat\beta_\orac)$ follows by combining the above two cases.

It suffices to prove \eqref{eq:amse_convergence_convex_comb}.
For notational simplicity, we write $\psi_\rct(X_i, A_i, Y_i) = \psi_{\rct, i}$ for $i\in\Rct$ and $\psi_\obs(X_i, A_i, Y_i) = \psi_{\obs, i}$ for $i\in\Obs$.
By Berry-Esseen theorem (see Theorem \ref{thm:berry_esseen}), for any fixed ${\tilde{\weight}}\in[0, 1]$, we have
\begin{align*}
     \sup_{t\in\bbR}\bigg|\bbP\bigg[ 
    \frac{ {(1-{\tilde{\weight}}) n_\rct^{-1} \sum_{i\in\Rct} \psi_{\rct, i}} + {{\tilde{\weight}} n_{\obs}^{-1}\sum_{i\in\Obs} \psi_{\obs, i}} }{\sqrt{\frac{(1-{\tilde{\weight}})^2 \sigma_\rct^2}{n_\rct} + \frac{{\tilde{\weight}}^2 \sigma_\obs^2}{n_{\obs}}}} \leq t\bigg] - \Phi(t)\bigg| 
    & \leq 6 \cdot \frac{\frac{(1-{\tilde{\weight}})^3 \rho_\rct}{n_\rct^2} + \frac{{\tilde{\weight}}^3 \rho_\obs}{n_{\obs}^2}}{\bigg(\frac{(1-{\tilde{\weight}})^2 \sigma_\rct^2}{n_\rct} + \frac{{\tilde{\weight}}^2\sigma_\obs^2}{n_{\obs}}\bigg)^{3/2}}\\
    & \leq 6 \bigg(\frac{\rho_\rct}{\sigma_\rct^3 \sqrt{n_\rct}} + \frac{\rho_\obs}{\sigma_\obs^3 \sqrt{n_{\obs}}}\bigg).
\end{align*}
Taking the supremum over ${\tilde{\weight}}\in[0, 1]$ at both sides, we get
\begin{align*}
    & \sup_{{\tilde{\weight}} \in [0, 1]}
    \sup_{t\in\bbR}\bigg|\bbP\bigg[ 
    \frac{ {(1-{\tilde{\weight}}) n_\rct^{-1} \sum_{i\in\Rct} \psi_{\rct, i}} + {{\tilde{\weight}} n_{\obs}^{-1}\sum_{i\in\Obs} \psi_{\obs, i}} }{\sqrt{\frac{(1-{\tilde{\weight}})^2 \sigma_\rct^2}{n_\rct} + \frac{{\tilde{\weight}}^2 \sigma_\obs^2}{n_{\obs}}}} \leq t\bigg] - \Phi(t)\bigg| 
    \label{eq:berry_esseen_comvex_comb}
    \leq 6 \bigg(\frac{\rho_\rct}{\sigma_\rct^3 \sqrt{n_\rct}} + \frac{\rho_\obs}{\sigma_\obs^3 \sqrt{n_{\obs}}}\bigg).\numberthis
\end{align*}
From the definition of $\hat\beta_{\tilde{\weight}}$ and Assumption \ref{assump:asymp_linear}, we have
\begin{align*}
    \frac{\hat\beta_{\tilde{\weight}} - [(1-{\tilde{\weight}})\beta^\star + {\tilde{\weight}} \beta_\obs]}{\sqrt{ \frac{(1-{\tilde{\weight}})^2 \sigma_\rct^2}{n_\rct} + \frac{{\tilde{\weight}}^2 \sigma_\obs^2}{n_{\obs}}}} 
    & = \frac{ (1-{\tilde{\weight}})n_\rct^{-1} \sum_{i\in\calC} \psi_{\rct, i} + {\tilde{\weight}} n_{\obs}^{-1}\sum_{i\in\Obs}\psi_{\obs, i} + (1-{\tilde{\weight}}) o_p(n_\rct^{-1/2}) + {\tilde{\weight}} o_p(n_{\obs}^{-1/2}) }{\sqrt{ \frac{(1-{\tilde{\weight}})^2 \sigma_\rct^2}{n_\rct} + \frac{{\tilde{\weight}}^2 \sigma_\obs^2}{n_{\obs}}}} \\
    & = \frac{ (1-{\tilde{\weight}})n_\rct^{-1} \sum_{i\in\calC} \psi_{\rct, i} + {\tilde{\weight}} n_{\obs}^{-1}\sum_{i\in\Obs}\psi_{\obs, i}  }{\sqrt{ \frac{(1-{\tilde{\weight}})^2 \sigma_\rct^2}{n_\rct} + \frac{{\tilde{\weight}}^2 \sigma_\obs^2}{n_{\obs}}}} + o_p(1).
\end{align*}
Using \eqref{eq:berry_esseen_comvex_comb} and invoking Slutsky's theorem, we get
\begin{equation}
    \label{eq:asymp_normal_convex_comb}
    \frac{\hat\beta_{\tilde{\weight}} - [(1-{\tilde{\weight}})\beta^\star + {\tilde{\weight}} \beta_\obs]}{\sqrt{ \frac{(1-{\tilde{\weight}})^2 \sigma_\rct^2}{n_\rct} + \frac{{\tilde{\weight}}^2 \sigma_\obs^2}{n_{\obs}}}} \xrightarrow{d} N(0, 1)
\end{equation}  
as $n_\rct \land n_{\obs} \to \infty$, where $\overset{d}{\to}$ denotes convergence in distribution and the rate of convergence is independent of the value of ${\tilde{\weight}}$. We now write
\begin{align*}
    & |\hat \beta_{\tilde{\weight}} - \beta^\star|^2 \\
    & = \bigg| (1-{\tilde{\weight}})\hat\beta_\rct + {\tilde{\weight}} \hat\beta_\obs - [(1-{\tilde{\weight}})\beta^\star + {\tilde{\weight}} \beta_\obs] - {\tilde{\weight}} \Delta\bigg|^2 \\
    & = \bigg| (1-{\tilde{\weight}})\hat\beta_\rct + {\tilde{\weight}} \hat\beta_\obs - [(1-{\tilde{\weight}})\beta^\star + {\tilde{\weight}} \beta_\obs]\bigg|^2 - 2 {\tilde{\weight}} \Delta \bigg((1-{\tilde{\weight}})\hat\beta_\rct + {\tilde{\weight}} \hat\beta_\obs - [(1-{\tilde{\weight}})\beta^\star + {\tilde{\weight}} \beta_\obs]\bigg) + {\tilde{\weight}}^2\Delta^2.
\end{align*}
By \eqref{eq:asymp_normal_convex_comb}, we have
\begin{align*}
    & \frac{\bbE\bigg| (1-{\tilde{\weight}})\hat\beta_\rct + {\tilde{\weight}} \hat\beta_\obs - [(1-{\tilde{\weight}})\beta^\star + {\tilde{\weight}} \beta_\obs]\bigg|^2}{ \frac{(1-{\tilde{\weight}})^2 \sigma_\rct^2}{n_\rct} + \frac{{\tilde{\weight}}^2 \sigma_\obs^2}{n_{\obs}}}  = 1+o(1), \qquad \frac{\bbE\bigg[ (1-{\tilde{\weight}})\hat\beta_\rct + {\tilde{\weight}} \hat\beta_\obs - [(1-{\tilde{\weight}})\beta^\star + {\tilde{\weight}} \beta_\obs] \bigg]}{\sqrt{ \frac{(1-{\tilde{\weight}})^2 \sigma_\rct^2}{n_\rct} + \frac{{\tilde{\weight}}^2 \sigma_\obs^2}{n_{\obs}}}} = o(1),
\end{align*}
where the $o(1)$ term is independent of the value of ${\tilde{\weight}}$. Thus, 
\begin{align*}
     \frac{\bbE|\hat\beta_{\tilde{\weight}} - \beta^\star|^2}{\frac{(1-{\tilde{\weight}})^2 \sigma_\rct^2}{n_\rct} + \frac{{\tilde{\weight}}^2 \sigma_\obs^2}{n_{\obs}} + {\tilde{\weight}}^2\Delta^2} 
    & =  1 + o(1) +  \frac{\bbE\bigg[(1-{\tilde{\weight}})\hat\beta_\rct + {\tilde{\weight}} \hat\beta_\obs - [(1-{\tilde{\weight}})\beta^\star + {\tilde{\weight}} \beta_\obs]\bigg]}{\sqrt{\frac{(1-{\tilde{\weight}})^2 \sigma_\rct^2}{n_\rct} + \frac{{\tilde{\weight}}^2 \sigma_\obs^2}{n_{\obs}}}} \cdot \frac{2 {\tilde{\weight}} \Delta \sqrt{ \frac{(1-{\tilde{\weight}})^2 \sigma_\rct^2}{n_\rct} + \frac{{\tilde{\weight}}^2 \sigma_\obs^2}{n_{\obs}} }}{\frac{(1-{\tilde{\weight}})^2 \sigma_\rct^2}{n_\rct} + \frac{{\tilde{\weight}}^2 \sigma_\obs^2}{n_{\obs}} + {\tilde{\weight}}^2\Delta^2} \\
    & = 1+ o(1) + o(1) \cdot \frac{2 {\tilde{\weight}} \Delta \sqrt{ \frac{(1-{\tilde{\weight}})^2 \sigma_\rct^2}{n_\rct} + \frac{{\tilde{\weight}}^2 \sigma_\obs^2}{n_{\obs}} }}{\frac{(1-{\tilde{\weight}})^2 \sigma_\rct^2}{n_\rct} + \frac{{\tilde{\weight}}^2 \sigma_\obs^2}{n_{\obs}} + {\tilde{\weight}}^2\Delta^2}.
\end{align*}
The proof of \eqref{eq:amse_convergence_convex_comb} is concluded by noting that
$$
    \frac{2 |{\tilde{\weight}} \Delta| \sqrt{ \frac{(1-{\tilde{\weight}})^2 \sigma_\rct^2}{n_\rct} + \frac{{\tilde{\weight}}^2 \sigma_\obs^2}{n_{\obs}} }}{\frac{(1-{\tilde{\weight}})^2 \sigma_\rct^2}{n_\rct} + \frac{{\tilde{\weight}}^2 \sigma_\obs^2}{n_{\obs}} + {\tilde{\weight}}^2\Delta^2} \leq 1,
$$
which follows from $a^2 + b^2 \geq 2ab$ for any $a, b \in \bbR$.

\subsection{Proof of Theorem \ref{thm:lb}}\label{prf:thm:lb}
    We first prove
    \begin{equation}
        \label{eq:lb_1}
        \inf_{\hat\beta} \sup_{(\beta^\star, \beta_\obs)\in \calP_{\overline\Delta}} \bbE [|\hat \beta - \beta^\star|^2] 
        \gtrsim \frac{\sigma_\rct^2}{n_\rct} \land \bar\Delta^2 .
    \end{equation}  
    To show this, consider the prior distribution $\beta^\star = \delta v_\rct, \beta_\obs = \delta v_\obs$, where $v_\rct, v_\obs \in \{\pm 1\}$ are independent symmetric Rademacher random variables, and $\delta$ is a constant that satisfies $\delta\leq \bar\Delta/2$, whose exact value will be determined later. The requirement of $\delta \leq \bar\Delta/2$ is to make sure 
    $$
    |\beta^\star - \beta_\obs| = \delta |v_\rct - v_\obs|\leq \overline\Delta,
    $$ 
    so that the specified prior distribution is indeed placed on the parameter space $\calP_{\overline\Delta}$. 
    Now, we proceed by
    \begin{align*}
        & \inf_{\hat\beta} \sup_{(\beta^\star, \beta_\obs)\in \calP_{\overline\Delta}} \bbE [|\hat \beta - \beta^\star|^2] \\
        & \geq 
        \inf_{\hat\beta} \bbE_{v_c, v_o} \bbE [|\hat \beta - \delta v_\rct|^2] \\
        & = \delta^2 \inf_{\hat v_\rct\in\{\pm 1\}}\bbE_{v_\rct, v_\obs} \bbE[|\hat v_\rct - v_\rct|^2] \\
        & \gtrsim \delta^2 \inf_{\hat v_\rct\in\{\pm 1\}}\bbE_{v_\rct, v_\obs} \bbP(\hat v_\rct \neq v_\rct)  \\
        & \geq \delta^2 \bbE_{v_\obs,\{\bfX_i\}, \{A_i\}}\inf_{\hat v_\rct\in\{\pm 1\}}\bbP(\hat v_\rct \neq v_\rct ~|~ v_\obs, \{\bfX_i\}, \{A_i\})\\
        & = \delta^2 \bbE_{v_\obs,\{\bfX_i\}, \{A_i\}} \inf_{\hat v_\rct\in\{\pm 1\}} \frac{1}{2}\bigg(\bbP(\hat v_\rct = -1 ~|~ v_\rct = +1, v_\obs, \{\bfX_i\}, \{A_i\})\\
        \label{eq:lb_reduction}
        & \qquad\qquad\qquad\qquad\qquad\qquad\qquad  + \bbP(\hat v_\rct = +1 ~|~ v_\rct = -1, v_\obs,  \{\bfX_i\}, \{A_i\})\bigg) \numberthis,
    \end{align*}
    Note that the infimum term in the right-hand side above is the type-\RN{1} plus type-\RN{2} error of testing $H_0: v_\rct = +1$ v.s. $H_1: v_\rct = -1$ conditional on the knowledge of $v_\obs$, $\{A_i: i\in \Rct\cup \Obs\}$ and $\{\bfX_i: i\in \Rct\cup\Obs\}$. The likelihood function under $H_0$ is given by
    \begin{align*}
        L_0 & = \bigg(\prod_{i\in\Rct, A_i = 1} \frac{1}{\sqrt{2\pi \sigma_\rct^2}} \exp\{-(Y_i - \delta/2 - \bsgamma_\rct^\top \bfX_i)^2/(2\sigma_\rct^2)\}\bigg) \\
        & \qquad \times \bigg(\prod_{i\in\Rct, A_i = 0}\frac{1}{\sqrt{2\pi \sigma_\rct^2}} \exp\{-(Y_i + \delta/2 - \bsgamma_\rct^\top \bfX_i)^2/(2\sigma_\rct^2)\}\bigg) \\
        & \qquad \times \bigg(\prod_{i\in\Obs, A_i = 1} \frac{1}{\sqrt{2\pi \sigma_\obs^2}} \exp\{-(Y_i - \delta v_\obs/2 - \bsgamma_\obs^\top \bfX_i)^2/(2\sigma_\obs^2)\}\bigg) \\
        & \qquad \times
        \bigg(\prod_{i\in\Obs, A_i = 0} \frac{1}{\sqrt{2\pi \sigma_\obs^2}} \exp\{-(Y_i + \delta v_\obs/2 -\bsgamma_\obs^\top \bfX_i)^2/(2\sigma_\obs^2)\}\bigg),
    \end{align*}    
    and the likelihood function under $H_1$ is given by
    \begin{align*}
        L_1 & = \bigg(\prod_{i\in\Rct, A_i = 1} \frac{1}{\sqrt{2\pi \sigma_\rct^2}} \exp\{-(Y_i + \delta/2 - \bsgamma_\rct^\top \bfX_i)^2/(2\sigma_\rct^2)\}\bigg) \\
        & \qquad \times \bigg(\prod_{i\in\Rct, A_i = 0} \frac{1}{\sqrt{2\pi \sigma_\rct^2}} \exp\{-(Y_i - \delta/2 - \bsgamma_\rct^\top \bfX_i)^2/(2\sigma_\rct^2)\}\bigg) \\
        & \qquad \times \bigg(\prod_{i\in\Obs, A_i = 1} \frac{1}{\sqrt{2\pi \sigma_\obs^2}} \exp\{-(Y_i - \delta v_\obs/2 - \bsgamma_\obs^\top \bfX_i)^2/(2\sigma_\obs^2)\}\bigg) \\
        & \qquad \times
        \bigg(\prod_{i\in\Obs, A_i = 0} \frac{1}{\sqrt{2\pi \sigma_\obs^2}} \exp\{-(Y_i + \delta v_\obs/2 -\bsgamma_\obs^\top \bfX_i)^2/(2\sigma_\obs^2)\}\bigg),
    \end{align*}    
    By Neyman--Pearson Lemma, the optimal test that minimizes the sum of type-\RN{1} and type-\RN{2} error is given by rejecting $H_0$ when $L_1 \geq L_0$. This criterion translates to rejecting $H_0$ when
    $$
        \sum_{i\in\Rct, A_i = 1} (Y_i -\bsgamma_\rct^\top \bfX_i) - \sum_{i\in\Rct, A_i = 0} (Y_i -\bsgamma_\rct^\top \bfX_i)\leq 0.
    $$
    Under $H_0$, we know that $Y_i - \bsgamma_\rct^\top \bfX_i \mid A_i, \bfX_i \overset{i.i.d.}{\sim} N( (A_i-1/2)\delta, \sigma_\rct^2)$ for any $i\in\Rct$.
    Denoting $n_{\rct, 1} = |\{i\in \Rct: A_i = 1\}|$ and $n_{\rct, 0} = |\{i\in \Rct: A_i = 0\}|$.
    The left-hand side above is then distributed as
    $$
        N\bigg(
            \frac{\delta(n_{c, 1} + n_{c, 0})}{2},
            \sigma_\rct^2 (n_{c, 1} + n_{c, 0})
        \bigg)
        \overset{d}{=}
        N\bigg(
            \frac{\delta n_\rct}{2},
            \sigma_\rct^2 n_\rct
        \bigg).
    $$
    Thus, the type-\RN{1} error is given by
    \begin{align*}
        \bbP\bigg[
        N\bigg(
            \frac{\delta n_\rct}{2},
            \sigma_\rct^2 n_\rct
        \bigg) \leq 0
        \bigg] = \Phi \bigg(- \frac{\delta}{2 \sigma_\rct/\sqrt{n_\rct}}\bigg).
    \end{align*}
    A symmetric argument shows that the type-\RN{2} error also admits the above lower bound. Thus, we have
    \begin{align*}
        \inf_{\hat\beta} \sup_{(\beta^\star, \beta_\obs)\in \calP_{\overline\Delta}} \bbE [|\hat \beta - \beta^\star|^2] 
        & \gtrsim \delta^2 \Phi\bigg(-\frac{\delta}{2 \sigma_{\rct}/\sqrt{n_{\rct}}}\bigg)
    \end{align*}
    for any $\delta\leq \overline\Delta/2$. We consider the following two cases.
    \begin{enumerate}
        \item If $\overline\Delta^2 \geq \sigma_\rct^2/n_\rct$, we choose $\delta = \sigma_\rct/(2\sqrt{n_\rct})$, which gives 
        $$
        \inf_{\hat\beta} \sup_{(\beta^\star, \beta_\obs)\in \calP_{\overline\Delta}} \bbE [|\hat \beta - \beta^\star|^2] 
        \gtrsim \frac{\sigma_\rct^2}{n_\rct}.
        $$
        \item If $\overline\Delta^2 < \sigma_\rct^2/n_\rct$, we choose $\delta = \overline\Delta/2$, which gives 
        $$
        \inf_{\hat\beta} \sup_{(\beta^\star, \beta_\obs)\in \calP_{\overline\Delta}} \bbE [|\hat \beta - \beta^\star|^2] 
        \gtrsim \overline\Delta^2\Phi\bigg(-\frac{\overline\Delta}{2 \sigma_\rct/\sqrt{n_c}}\bigg) \gtrsim \overline\Delta^2 .
        $$
    \end{enumerate}
    Equation \eqref{eq:lb_1} follows from the above two cases. 

    We then show
    \begin{equation}
        \label{eq:lb_2}
        \inf_{\hat\beta} \sup_{(\beta^\star, \beta_\obs)\in \calP_{\overline\Delta}} \bbE [|\hat \beta - \beta^\star|^2] 
        \gtrsim \frac{\sigma_\rct^2 \sigma_\obs^2}{n_\rct \sigma_\obs^2 + n_\obs \sigma_\rct^2}.
    \end{equation}  
    If the above inequality holds, then we get
    $$
        \inf_{\hat\beta} \sup_{(\beta^\star, \beta_\obs)\in \calP_{\overline\Delta}} \bbE [|\hat \beta - \beta^\star|^2] 
        \gtrsim \bigg(\frac{\sigma_\rct^2 \sigma_\obs^2}{n_\rct \sigma_\obs^2 + n_\obs \sigma_\rct^2}\bigg) \lor  \bigg(\frac{\sigma_\rct^2}{n_\rct} \land \overline\Delta^2\bigg) 
        \gtrsim \frac{\sigma_\rct^2 \sigma_\obs^2}{n_\rct \sigma_\obs^2 + n_\obs \sigma_\rct^2} + \frac{\sigma_\rct^2}{n_\rct} \land \overline\Delta^2
        ,
    $$
    which is our desired result. To show \eqref{eq:lb_2}, we consider the prior distribution $\beta^\star = \beta_\obs = \delta v$, where $v$ is a symmetric Rademacher random variable and $\delta > 0$ is a constant whose value will be specified later. 
    Note that under this prior distribution, $|\beta^\star-\beta_\obs| = 0$, so $(\beta^\star, \beta_\obs)\in\calP(\bar\Delta)$ automatically holds. Using a nearly identical argument as what led to \eqref{eq:lb_reduction}, we have
    \begin{align*}
        & \inf_{\hat\beta} \sup_{(\beta^\star, \beta_\obs)\in \calP_{\overline\Delta}} \bbE [|\hat \beta - \beta^\star|^2] \\
        & \gtrsim \delta^2 \bbE_{\{\bfX_i\}, \{A_i\}}\inf_{\hat v\in\{\pm 1\}} \frac{1}{2}\bigg(\bbP(\hat v= -1 ~|~ v = +1, \{A_i\}, \{\bfX_i\}) + \bbP(\hat v= +1 ~|~ v= -1, \{A_i\}, \{\bfX_i\})\bigg) ,
    \end{align*}
    The right-hand side above again amounts to a binary hypothesis testing problem between $H_0: v = +1$ and $H_1: v = -1$ given the knowledge of $\{A_i\}$ and $\{\bfX_i\}$. The likelihood functions can again be explicitly calculated, and with some algebra, we arrive at the following optimal test, which rejects $H_0$ when
    $$
        \frac{1}{\sigma_\rct^2} \sum_{i\in\Rct, A_i = 1} (Y_i - \bsgamma_\rct^\top \bfX_i) - \frac{1}{\sigma_\rct^2} \sum_{i\in\Rct, A_i = 0} (Y_i - \bsgamma_\rct^\top \bfX_i) + \frac{1}{\sigma_\obs^2} \sum_{i\in\Obs, A_i = 1} (Y_i - \bsgamma_\obs^\top \bfX_i) -  \frac{1}{\sigma_\obs^2} \sum_{i\in\Obs, A_i = 0} (Y_i - \bsgamma_\obs^\top \bfX_i) \leq 0.
    $$
    Under $H_0$, we have $Y_i- \bsgamma_\rct^\top \bfX_i \mid A_i, \bfX_i \overset{i.i.d.}\sim N( (A_i-1/2)\delta, \sigma_\rct^2)$ for $i \in \Rct$ and $Y_i - \bsgamma_\obs^\top \bfX_i\mid A_i, \bfX_i \overset{i.i.d.}\sim N((A_i - 1/2)\delta, \sigma_\obs^2)$ for $i \in \Obs$. Thus, 
    the type-\RN{1} error (and type-\RN{2} error as well by a symmetric argument) takes the following form:
    \begin{align*}
        \bbP\bigg[ N\bigg( \frac{\delta n_\rct}{2\sigma_\rct^2} + \frac{\delta n_{\obs}}{2\sigma_\obs^2}, \frac{n_{\rct} }{\sigma_\rct^2} + \frac{n_{\obs}}{\sigma_\obs^2}\bigg)  \leq 0\bigg] 
        & = \Phi\bigg({-\frac{\delta}{2}\sqrt{\frac{n_{\rct}}{\sigma_\rct^2} + \frac{n_{\obs}}{2\sigma_\obs^2}}}\bigg) .
    \end{align*}
    Choosing $\delta = 1 / \sqrt{\frac{n_{\rct}}{\sigma_\rct^2} + \frac{n_{\obs}}{\sigma_\obs^2}}$, we get
    $$
        \inf_{\hat\beta} \sup_{(\beta^\star, \beta_\obs)\in \calP_{\overline\Delta}} \bbE [|\hat \beta - \beta^\star|^2] \gtrsim \frac{1}{n_\rct/\sigma_\rct^2 + n_\obs/\sigma_\obs^2},
    $$
    which is exactly \eqref{eq:lb_2}. The proof is concluded.

\subsection{Proof of Theorem \ref{thm:high_prob_bound_adaptive_estimator}} \label{prf:thm:high_prob_bound_adaptive_estimator}

Similar to the proof of Theorem \ref{thm:oracle_estimator}, we write $\psi_\rct(X_i, A_i, Y_i) = \psi_{\rct, i}$ for $i\in\Rct$ and $\psi_\obs(X_i, A_i, Y_i) = \psi_{\obs, i}$ for $i\in\Obs$.

We first show that $\Delta$ is feasible for \eqref{eq:est_delta} with high probability. By construction, we have
\begin{align*}
    \frac{\tilde \Delta - \Delta}{\sqrt{\frac{\sigma_\rct^2}{n_\rct} + \frac{\sigma_\obs^2}{n_{\obs}}}} 
    = \frac{n_\rct^{-1}\sum_{i\in\Rct} \psi_{\rct, i} - n_{\obs}^{-1}\sum_{i\in\Obs}\psi_{\obs, i}}{\sqrt{\frac{\sigma_\rct^2}{n_\rct} + \frac{\sigma_\obs^2}{n_{\obs}}}} + o_p(1).
\end{align*}
An application of Berry-Esseen theorem (see Theorem \ref{thm:berry_esseen}) gives
\begin{align*}
    \sup_{t\in\bbR}\bigg| \bbP\bigg( \frac{n_\rct^{-1}\sum_{i\in\Rct} \psi_{\rct, i} - n_{\obs}^{-1}\sum_{i\in\Obs}\psi_{\obs, i}}{\sqrt{\frac{\sigma_\rct^2}{n_\rct} + \frac{\sigma_\obs^2}{n_{\obs}}}} \leq t \bigg) -\Phi(t) \bigg| \leq 6 \cdot \frac{\frac{\rho_\rct}{n_\rct^2} + \frac{\rho_\obs}{n_{\obs}^2}}{(\frac{\sigma_\rct^2}{n_\rct} + \frac{\sigma_\obs^2}{n_{\obs}})^{3/2}} \leq 6 \bigg(\frac{\rho_\rct}{\sigma_\rct^3 \sqrt{n_\rct}} + \frac{\rho_\obs}{\sigma_\obs^3 \sqrt{n_{\obs}}}\bigg)  = o(1),
\end{align*}
where the last inequality is by $n_\rct \land n_{\obs} \to \infty$.
Thus, for any $t > 0$, we have
$$
    \bigg|\frac{n_\rct^{-1}\sum_{i\in\Rct} \psi_{\rct, i} - n_{\obs}^{-1}\sum_{i\in\Obs}\psi_{\obs, i}}{\sqrt{\frac{\sigma_\rct^2}{n_\rct} + \frac{\sigma_\obs^2}{n_{\obs}}}}\bigg| \leq t
$$
with probability at least $1-2\Phi(-t) - o(1)$. As long as $t \gtrsim 1$, the magnitude of the $o_p(1)$ term is at most $t$ with probability at least $1-o(1)$. Thus, a union bound gives that with probability at least $1-2\Phi(-t) - o(1)$, 
$$
    |\tilde\Delta - \Delta| \lesssim  t \sqrt{\frac{\sigma_\rct^2}{n_\rct} + \frac{\sigma_\obs^2}{n_{\obs}}}.
$$
Consider the event
$$
    \calE := \{0.9\sigma_\rct \leq \hat \sigma_\rct \leq 1.1\sigma_\rct\} \cap \{0.9\sigma_\obs \leq \hat \sigma_\obs \leq 1.1\sigma_\obs\}.
$$
Since $\hat\sigma_\rct\overset{p}{\to}{\sigma_\rct}$ and $\hat\sigma_\obs\overset{p}{\to}{\sigma_\obs}$, we know that $\bbP(\calE) \geq 1-o(1)$. Thus, we have
$$
    |\tilde\Delta - \Delta| \leq C t \sqrt{\frac{\hat\sigma_\rct^2}{n_\rct} + \frac{\hat\sigma_\obs^2}{n_{\obs}}}
$$
with probability $1-2\Phi(-t)-o(1)$, where $C>0$ is some absolute constant. Under the choice of $t = \lambda/C$, the above bound reads
\begin{align*}
    \bbP\bigg(|\tilde\Delta - \Delta| \leq \lambda \cdot \sqrt{\frac{\hat\sigma_\rct^2}{n_\rct} + \frac{\hat\sigma_\obs^2}{n_{\obs}}}\bigg) 
    & \geq 1 - 2\Phi(-{\lambda}/{C}) - o(1)\\
    & \geq 1 - \calO(e^{-\lambda^2/(2C^2)}) - o(1) \\
    \label{eq:high_prob_feasible}
    & \geq 1- e^{-\sfc \lambda^2} - o(1).\numberthis
\end{align*}    
Here, the second inequality above is by the Gaussian tail bound $\Phi(-t) \leq \frac{1}{t\sqrt{2\pi}}e^{-t^2/2}$ and $\lambda \gtrsim 1$, and $\sfc > 0$ is an absolute constant only depending on $C$. Thus, we have shown that $\Delta$ is feasible for for \eqref{eq:est_delta} with probability at least $1- e^{-\sfc \lambda^2} - o(1)$.
On this high probability event, we have $|\hat\Delta_\lambda|\leq |\Delta|$, which further yields
\begin{equation}
    \label{eq:ub_by_delta}
    |\hat\Delta_\lambda -\Delta| \leq 2|\Delta|.
\end{equation}
Meanwhile, on the same high probability event, we have
\begin{align*}
    |\hat\Delta_\lambda - \Delta| & \leq |\tilde \Delta - \Delta| + |\hat \Delta_\lambda - \tilde \Delta|  \\
    & \leq \lambda \cdot \sqrt{\frac{\hat\sigma_\rct^2}{n_\rct} + \frac{\hat\sigma_\obs^2}{n_{\obs}}} +  |\hat\Delta_\lambda - \tilde \Delta|\\
    & \leq 2\lambda \cdot \sqrt{\frac{\hat\sigma_\rct^2}{n_\rct} + \frac{\hat\sigma_\obs^2}{n_{\obs}}},
\end{align*}    
where the second inequality is by \eqref{eq:high_prob_feasible} and the last inequality is by the constraint in \eqref{eq:est_delta}. Recall that the high probability event in \eqref{eq:high_prob_feasible} is a superset of the event $\calE$, so we can further upper bound the right-hand side above by
\begin{equation}
    \label{eq:ub_by_sd}
    2\sqrt{1.1} \cdot \lambda \cdot \sqrt{\frac{\sigma_\rct^2}{n_\rct} + \frac{\sigma_\obs^2}{n_{\obs}}}.
\end{equation}  
Summarizing \eqref{eq:ub_by_delta} and \eqref{eq:ub_by_sd}, we arrive at
\begin{equation}
    \label{eq:est_err_delta}
    \bbP\bigg[|\hat\Delta_\lambda - \Delta| \lesssim \bigg(\lambda \cdot \sqrt{\frac{\sigma_\rct^2}{n_\rct} + \frac{\sigma_\obs^2}{n_{\obs}}}\bigg) \land |\Delta|\bigg]  \geq 1- e^{-\sfc \lambda^2} - o(1).
\end{equation}  
Denote $\tilde \beta = (1-\hat\weight)\hat\beta_\rct + \hat\weight \hat\beta_\obs$, so that $\hat\beta_\lambda = \tilde \beta + \hat\weight \hat\Delta_\lambda$.
To analyze $\hat\beta_\lambda$, we decompose
\begin{align*}
    |\hat \beta_\lambda - \beta^\star|^2 & \lesssim \bigg|\tilde \beta - \bigg(\hat\weight \cdot \beta_\obs + (1-\hat\weight) \cdot \beta^\star\bigg)\bigg|^2 + \bigg|\bigg(\hat\weight \cdot \beta_\obs + (1-\hat\weight) \cdot \beta^\star\bigg) - \beta^\star + \omega\hat \Delta_\lambda\bigg|^2 \\
    & = \bigg|\tilde \beta - \bigg(\hat\weight \cdot \beta_\obs + (1-\hat\weight) \cdot \beta^\star\bigg)\bigg|^2 + \hat\weight^2|\hat \Delta_\lambda - \Delta|^2 \\
    \label{eq:param_est_err_decomp}
    & \lesssim \bigg|\tilde \beta - \bigg(\hat\weight \cdot \beta_\obs + (1-\hat\weight) \cdot \beta^\star\bigg)\bigg|^2 + \bigg(\lambda^2 \cdot {\frac{\sigma_\rct^2}{n_\rct}}\bigg) \land \Delta^2,\numberthis
\end{align*}
where the last inequality happens with probability $1-e^{-\sfc \lambda^2} - o(1)$ by \eqref{eq:est_err_delta} and $\frac{\sigma_\rct^2}{n_\rct}\gtrsim \frac{\sigma_\obs^2}{n_{\obs}}$. On the event $\calE$, we have
\begin{align*}
    \hat\weight & =  \frac{\hat \sigma_\rct^2/n_\rct}{\hat \sigma_\rct^2/n_\rct + \hat\sigma_\obs^2/n_{\obs}} \in \Big[\frac{0.9\weight}{1.1} , \frac{1.1\weight}{0.9}\Big],\\
    1 - \hat\weight &  =  \frac{\hat \sigma_\obs^2/n_{\obs}}{\hat \sigma_\rct^2/n_\rct + \hat\sigma_\obs^2/n_{\obs}} \in \Big[\frac{0.9(1-\weight)}{1.1} , \frac{1.1(1-\weight)}{0.9}\Big].
\end{align*}    
Thus, on the event $\calE$, we can proceed by
\begin{align*}
    \bigg|\tilde \beta - \bigg(\hat\weight \cdot \beta_\obs + (1-\hat\weight) \cdot \beta^\star\bigg)\bigg|^2 
    & \lesssim (1-\weight)^2 |\hat\beta_\rct - \beta^\star|^2 + \weight^2 |\hat\beta_\obs - \beta_\obs|^2.
\end{align*}
Using Assumption \ref{assump:asymp_linear}, we write
\begin{align*}
    (1-\weight)|\hat\beta_\rct- \beta^\star| = (1-\weight) \cdot \bigg| \frac{1}{n_\rct}\sum_{i\in\Rct}\psi_{\rct, i} + o_p(n_\rct^{-1/2})\bigg|.
\end{align*}
By Berry-Esseen theorem (see Theorem \ref{thm:berry_esseen}), we have
$$
    \sup_{t\in\bbR} \bigg| \bbP\bigg(\frac{1}{n_\rct}\sum_{i\in\Rct}\psi_{\rct, i} \leq t \cdot \frac{\sigma_\rct}{\sqrt{n_\rct}}\bigg) - \Phi(t)\bigg| \leq \frac{6 \rho_\rct}{\sigma_\rct^3 \sqrt{n_\rct}},
$$
which further gives
$$
    \bbP\bigg(\frac{1}{n_\rct}\sum_{i\in\Rct}\psi_{\rct, i} \leq \lambda \cdot \frac{\sigma_\rct}{\sqrt{n_\rct}}\bigg) \geq 1 - \Phi(-\lambda)  - o(1) \geq 1 - \calO(e^{-\lambda^2/2}) - o(1),
$$
where the last inequality is by the Gaussian tail bound and $\lambda \gtrsim 1$.
The above display, along with $\lambda \gtrsim 1$, implies that with probability $1-e^{-\sfc \lambda^2} - o(1)$ for a potentially different absolute constant $\sfc >0$, we have
\begin{align*}
    (1-\weight)|\hat\beta_\rct- \beta^\star| 
    & \lesssim \lambda \cdot \frac{\sigma_\obs^2/n_{\obs}}{\sigma_\rct^2/n_\rct + \sigma_\obs^2/n_{\obs}} \cdot \frac{\sigma_\rct}{\sqrt{n_\rct}} \\
    & \leq \lambda \cdot \frac{\sigma_\obs^2/n_{\obs}}{\sigma_\rct^2/n_\rct} \cdot \frac{\sigma_\rct}{\sqrt{n_\rct}} \\
    & = \lambda \cdot \frac{\sigma_\obs}{\sqrt{n_{\obs}}} \cdot \sqrt{\frac{\sigma_\obs^2/n_{\obs}}{\sigma_\rct^2/n_\rct}} \\
    & \lesssim \lambda \cdot \frac{\sigma_\obs}{\sqrt{n_{\obs}}} \\
    & \lesssim \lambda \cdot \sqrt{\frac{\sigma_\rct^2 \sigma_\obs^2}{n_\rct \sigma_\obs^2 + n_{\obs} \sigma_\rct^2}},
\end{align*}
where the last two inequalities are by $\sigma_\rct^2/n_\rct\gtrsim \sigma_\obs^2/n_{\obs}$. Using a nearly identical argument, we have
$$
    \bbP\bigg(\frac{1}{n_{\obs}}\sum_{i\in\Obs}\psi_{\obs, i} \leq \lambda \cdot \frac{\sigma_\obs}{\sqrt{n_{\obs}}}\bigg) \geq 1 - e^{-\sfc \lambda^2}-o(1),
$$
which further implies that,
\begin{align*}
    \weight |\hat\beta_\obs - \beta_\obs| 
    & \lesssim \lambda \cdot \frac{\sigma_\rct^2/n_\rct}{\sigma_\rct^2/n_\rct + \sigma_\obs^2/n_{\obs}} \cdot \frac{\sigma_\obs}{\sqrt{\sigma_\obs}} \\
    & \leq\lambda \cdot  \frac{\sigma_\obs}{\sqrt{\sigma_\obs}} \\
    & \lesssim \lambda \cdot \sqrt{\frac{\sigma_\rct^2 \sigma_\obs^2}{n_\rct \sigma_\obs^2 + n_{\obs} \sigma_\rct^2}},
\end{align*}
where the last inequality is again by $\sigma_\rct^2/n_\rct\gtrsim \sigma_\obs^2/n_{\obs}$. A union bound then gives
$$
    \bigg|\tilde \beta - \bigg(\hat\weight \cdot \beta_\obs + (1-\hat\weight) \cdot \beta^\star\bigg)\bigg|^2  \lesssim \lambda^2\cdot \frac{\sigma_\rct^2 \sigma_\obs^2}{n_\rct \sigma_\obs^2 + n_{\obs} \sigma_\rct^2}
$$
with probability $1-2 e^{-\sfc\lambda^2} - o(1)$. The proof is concluded by plugging the above display back to \eqref{eq:param_est_err_decomp}, invoking another union bound, and recalling that $\lambda\gtrsim 1$ and $|\Delta|\leq \overline\Delta$.

\subsection{Proof of Theorem \ref{thm:oracle_ci}} \label{prf:thm:oracle_ci}
We first show the proposed oracle CI is asymptotically level $(1-\alpha)$. This conclusion is trivial for the $\overline\Delta \geq \sigma_\rct/\sqrt{n_\rct}$ case. When $\overline\Delta < \sigma_\rct/\sqrt{n_\rct}$, recall that we have shown in the proof of Theorem \ref{thm:oracle_estimator} that
$$
    {\hat\beta_\weight}\bigg/\bigg({\frac{\sigma_\rct\sigma_\obs}{\sqrt{n_\rct\sigma_\obs^2 + n_\obs\sigma_\rct^2}} }\bigg) \overset{d}{\to} N(0, 1).
$$
Hence, we have
$$
    \bbP\bigg( - \Phi^{-1}(1-\frac{\alpha}{2})\cdot \frac{\sigma_\rct\sigma_\obs}{\sqrt{n_\rct\sigma_\obs^2 + n_{\obs}\sigma_\rct^2}} + \hat\beta_\weight \leq (1-\weight)\beta^\star + \weight \beta_\obs \leq \hat\beta_\weight + \Phi^{-1}(1-\frac{\alpha}{2})\cdot \frac{\sigma_\rct\sigma_\obs}{\sqrt{n_\rct\sigma_\obs^2 + n_{\obs}\sigma_\rct^2}}\bigg) \geq 1 - \alpha - o(1).
$$
Since $\beta^\star = (1-\weight)\beta^\star + \weight \beta_\obs + \weight \Delta$ and $|\Delta|\leq \overline\Delta$, it follows that
\begin{align*}
    & \bbP\bigg( -\overline\Delta- \Phi^{-1}(1-\frac{\alpha}{2})\cdot \frac{\sigma_\rct\sigma_\obs}{\sqrt{n_\rct\sigma_\obs^2 + n_{\obs}\sigma_\rct^2}} + \hat\beta_\weight 
    \leq \beta^\star 
    \leq \hat\beta_\weight + \Phi^{-1}(1-\frac{\alpha}{2})\cdot \frac{\sigma_\rct\sigma_\obs}{\sqrt{n_\rct\sigma_\obs^2 + n_{\obs}\sigma_\rct^2}} + \overline \Delta\bigg) \\
    & \geq \bbP\bigg( - \Phi^{-1}(1-\frac{\alpha}{2})\cdot \frac{\sigma_\rct\sigma_\obs}{\sqrt{n_\rct\sigma_\obs^2 + n_{\obs}\sigma_\rct^2}} + \hat\beta_\weight \leq (1-\weight)\beta^\star + \weight \beta_\obs \leq \hat\beta_\weight + \Phi^{-1}(1-\frac{\alpha}{2})\cdot \frac{\sigma_\rct\sigma_\obs}{\sqrt{n_\rct\sigma_\obs^2 + n_{\obs}\sigma_\rct^2}}\bigg)\\
    & \geq 1 - \alpha - o(1).
\end{align*} 
By construction, 
\begin{align*}
    \len([L_\alpha, U_\alpha])
     = 
    \begin{cases}
        2 \Phi^{-1}(1-\alpha/2) \frac{\sigma_\rct}{\sqrt{n_\rct}} & \textnormal{ if } \overline\Delta \geq \sigma_\rct/\sqrt{n_\rct} \\
        2 \Phi^{-1}(1-\alpha/2) \frac{\sigma_\rct\sigma_\obs}{\sqrt{n_\rct\sigma_\obs^2 + n_{\obs}\sigma_\rct^2}} + 2 \overline\Delta  & \textnormal{ otherwise}.
    \end{cases}
\end{align*}
As long as $\alpha \geq \ep\gtrsim 1$, we have $\Phi^{-1}(1-\alpha/2)\lesssim 1$. 
It follows that  
$$
    \len([L_\alpha, U_\alpha]) \lesssim \frac{\sigma_\rct\sigma_\obs}{\sqrt{n_\rct\sigma_\obs^2 + n_\obs\sigma_\rct^2}} + \frac{\sigma_\rct}{\sqrt{n_\rct}} \land \overline\Delta,
$$
and the proof is concluded by taking expectation at both sides.

\subsection{Proof of Theorem \ref{thm:minimax_adaptive_length}}\label{prf:thm:minimax_adaptive_length}
We first present two useful lemmas.
\begin{lemma}
    \label{lemma:minimax_length}
    Let $\pi_{\underline\Delta}$ and $\pi_{\overline\Delta}$ be prior distributions on $\calP_{\underline\Delta}$ and $\calP_{\overline\Delta}$, respectively. Let $f_{\pi_{\underline\Delta}}$ and $f_{\pi_{\overline\Delta}}$ be the marginal distribution of all the data, marginalized over $\pi_{\underline\Delta}$ and $\pi_{\overline\Delta}$, respectively.
    Assume the prior distributions are chosen so that $\beta^\star = \underline\beta$ under $\pi_{\underline\Delta}$ and $\beta^\star = \overline\beta$ under $\pi_{\overline\Delta}$. Then we have
    $$
        \len^\star_\alpha(\calP_{\underline\Delta}, \calP_{\overline\Delta}) \geq |\underline\beta - \overline\beta| \cdot \bigg[\bigg(1 - TV(f_{\pi_{\underline\Delta}}, f_{\pi_{\overline\Delta}}) - 2\alpha\bigg) \lor 0\bigg],
    $$
    where $TV(\cdot, \cdot)$ is the total variation distance.
\end{lemma}
\begin{proof}
    This is a direct corollary of Lemma 1 in \cite{cai2017confidence}.
\end{proof}

\begin{lemma}
    \label{lemma:tv_ineq}
    Let $X, Y, Z$ be random variables. Then
    $$
        TV(\textnormal{Law}(X), \textnormal{Law}(Y)) \leq \bbE_{Z}[TV(\textnormal{Law}(X|Z), \textnormal{Law}(Y|Z))].
    $$
\end{lemma}
\begin{proof}
    By definition, we have
    \begin{align*}
        TV(\textnormal{Law}(X), \textnormal{Law}(Y)) & = \sup_{\calE}|\bbP(X\in \calE) - \bbP(Y \in \calE)|\\
        & = \sup_{\calE} \bigg|\bbE_Z \bigg[\bbP(X\in\calE \mid Z) - \bbP(Y\in\calE \mid Z)\bigg] \bigg|\\
        & \leq \sup_{\calE} \bbE_Z \bigg| \bbP(X\in\calE \mid Z) - \bbP(Y\in\calE \mid Z)\bigg| \\
        & \leq \bbE_Z \sup_{\calE}\bigg| \bbP(X\in\calE \mid Z) - \bbP(Y\in\calE \mid Z)\bigg| \\
        & = \bbE_{Z}[TV(\textnormal{Law}(X|Z), \textnormal{Law}(Y|Z))],
    \end{align*}
    where the last two inequalities are by Jensen's inequality.
\end{proof}

We break the proof into three cases.

\paragraph{Case A. $\boldsymbol{\overline\Delta \geq \sigma_\rct/\sqrt{n_\rct}}$.}
In this case, we choose $\pi_{\underline\Delta}$ to be a point mass at $\beta^\star = 0, \beta_\obs = 0$ and we let $\pi_{\overline\Delta}$ be a point mass at $\beta^\star = \sigma_\rct/\sqrt{n_\rct}, \beta_\obs = 0$. By Lemmas \ref{lemma:minimax_length} and \ref{lemma:tv_ineq}, we have
\begin{align*}
    \len^\star_\alpha(\calP_{\underline\Delta}, \calP_{\overline\Delta}) 
    &\geq \frac{\sigma_\rct}{\sqrt{n_\rct}} \cdot \bigg[\bigg(1 - TV(f_{\pi_{\underline\Delta}}, f_{\pi_{\overline\Delta}}) - 2\alpha\bigg) \lor 0\bigg] \\
    & \geq \frac{\sigma_\rct}{\sqrt{n_\rct}} \cdot \bigg[\bigg( \bbE_{\{\bfX_i\}, \{A_i\}}[ 1 - TV(f_{\pi_{\underline\Delta}|\{\bfX_i\}, \{A_i\}}, f_{\pi_{\overline\Delta}|\{\bfX_i\}, \{A_i\}})] - 2\alpha\bigg) \lor 0\bigg],
\end{align*}
where $f_{\pi_{\underline\Delta}|\{\bfX_i\}, \{A_i\}}$ (resp.~$f_{\pi_{\overline\Delta}|\{\bfX_i\}, \{A_i\}}$) is the conditional distribution of $f_{\pi_{\underline\Delta}}$ (resp.~$f_{\pi_{\underline\Delta}}$) given $\{\bfX_i\}$ and $\{A_i\}$. 
By Neyman--Pearson lemma, the quantity $1 - TV(f_{\pi_{\underline\Delta}|\{\bfX_i\}, \{A_i\}}, f_{\pi_{\overline\Delta}|\{\bfX_i\}, \{A_i\}})$ is equal to the Type-I plus Type-II error of testing
$$
    H_0: \{Y_i\}\mid \{\bfX_i\}, \{A_i\} \sim f_{\pi_{\underline\Delta}|\{\bfX_i\}, \{A_i\}}
    ~~~
    \textnormal{v.s.}
    ~~~
    H_1: \{Y_i\}\mid \{\bfX_i\}, \{A_i\} \sim f_{\pi_{\overline\Delta}|\{\bfX_i\}, \{A_i\}},
$$
and the optimal test is given by the likelihood ratio test.
By construction, the likelihood function under $H_0$ is
\begin{align*}
    L_0 & = \bigg(\prod_{i\in\Rct, A_i = 1} \frac{1}{\sqrt{2\pi \sigma_\rct^2}} \exp\{-(Y_i - \bsgamma_\rct^\top \bfX_i)^2/(2\sigma_\rct^2)\}\bigg) \\
    & \qquad \times \bigg(\prod_{i\in\Rct, A_i = 0}\frac{1}{\sqrt{2\pi \sigma_\rct^2}} \exp\{-(Y_i - \bsgamma_\rct^\top \bfX_i)^2/(2\sigma_\rct^2)\}\bigg) \\
    & \qquad \times \bigg(\prod_{i\in\Obs, A_i = 1} \frac{1}{\sqrt{2\pi \sigma_\obs^2}} \exp\{-(Y_i -  \bsgamma_\obs^\top \bfX_i)^2/(2\sigma_\obs^2)\}\bigg) \\
    \label{eq:minimax_length_L0_caseA}
    & \qquad \times
    \bigg(\prod_{i\in\Obs, A_i = 0} \frac{1}{\sqrt{2\pi \sigma_\obs^2}} \exp\{-(Y_i -\bsgamma_\obs^\top \bfX_i)^2/(2\sigma_\obs^2)\}\bigg),\numberthis
\end{align*}    
whereas the likelihood function under $H_1$ is
\begin{align*}
    L_1 & = \bigg(\prod_{i\in\Rct, A_i = 1} \frac{1}{\sqrt{2\pi \sigma_\rct^2}} \exp\{-(Y_i - \sigma_\rct/(2\sqrt{n_\rct}) - \bsgamma_\rct^\top \bfX_i)^2/(2\sigma_\rct^2)\}\bigg) \\
    & \qquad \times \bigg(\prod_{i\in\Rct, A_i = 0}\frac{1}{\sqrt{2\pi \sigma_\rct^2}} \exp\{-(Y_i + \sigma_\rct/(2\sqrt{n_\rct}) - \bsgamma_\rct^\top \bfX_i)^2/(2\sigma_\rct^2)\}\bigg) \\
    & \qquad \times \bigg(\prod_{i\in\Obs, A_i = 1} \frac{1}{\sqrt{2\pi \sigma_\obs^2}} \exp\{-(Y_i -  \bsgamma_\obs^\top \bfX_i)^2/(2\sigma_\obs^2)\}\bigg) \\
    & \qquad \times
    \bigg(\prod_{i\in\Obs, A_i = 0} \frac{1}{\sqrt{2\pi \sigma_\obs^2}} \exp\{-(Y_i -\bsgamma_\obs^\top \bfX_i)^2/(2\sigma_\obs^2)\}\bigg).
\end{align*}    
The optimal test rejects $H_0$ in favor of $H_1$ when $L_1/L_0\geq 1$, which is equivalent to
\begin{align*}
    -\sum_{i\in\calC, A_i = 1} \bigg(2(Y_i - \gamma_\rct^\top \bfX_i) - \frac{\sigma_\rct}{2\sqrt{n_\rct}}\bigg) 
    + 
    \sum_{i\in\calC, A_i = 0} \bigg(2(Y_i - \gamma_\rct^\top \bfX_i) + \frac{\sigma_\rct}{2\sqrt{n_\rct}}\bigg)  \leq 0.
\end{align*} 
Under $H_0$, we have $Y_i-\gamma_\rct^\top \bfX_i \mid \bfX_i, A_i \overset{i.i.d.}{\sim} N(0, \sigma_\rct^2)$. Thus, the left-hand side above is distributed as
$
    N\left(\frac{\sigma_\rct \sqrt{n_\rct}}{2}, 4n_\rct\sigma_\rct^2\right),
$
and the type-I error is given by
$$
    \bbP\bigg( N\left(\frac{\sigma_\rct \sqrt{n_\rct}}{2}, 4n\sigma_\rct^2\right) \leq 0\bigg) = \Phi(-1/4).
$$
A nearly identical argument shows that the Type-II error is also given by $\Phi(-1/4)$.
Hence, we get
$$
    1 - TV(f_{\pi_{\underline\Delta}|\{\bfX_i\}, \{A_i\}}, f_{\pi_{\overline\Delta}|\{\bfX_i\}, \{A_i\}})] = 2 \Phi(-1/4),
$$
and it follows that
$$
    \len^\star_\alpha(\calP_{\underline\Delta}, \calP_{\overline\Delta}) 
    \geq 2\frac{\sigma_\rct}{\sqrt{n_\rct}} \cdot (\Phi(-1/4) - \alpha)\gtrsim \frac{\sigma_\rct}{\sqrt{n_\rct}},
$$
where the last inequality is by $\alpha \leq \Phi(-1/4)-\ep$.

\paragraph{Case B. $\boldsymbol{ 1/\sqrt{n_\rct/\sigma_\rct^2 + n_\obs/\sigma_\obs^2}\leq \overline \Delta \leq \sigma_\rct/\sqrt{n_\rct}}$.}
In this case, we choose $\pi_{\underline\Delta}$ to be a point mass at $\beta^\star =0, \beta_\obs = 0$, and we choose $\pi_{\overline\Delta}$ to be a point mass at $\beta^\star = \overline\Delta, \beta_\obs = 0$. By Lemmas \ref{lemma:minimax_length} and \ref{lemma:tv_ineq}, we have
\begin{align*}
    \len^\star_\alpha(\calP_{\underline\Delta}, \calP_{\overline\Delta}) 
    & \geq \overline\Delta \cdot \bigg[\bigg( \bbE_{\{\bfX_i\}, \{A_i\}}[ 1 - TV(f_{\pi_{\underline\Delta}|\{\bfX_i\}, \{A_i\}}, f_{\pi_{\overline\Delta}|\{\bfX_i\}, \{A_i\}})] - 2\alpha\bigg) \lor 0\bigg].
\end{align*}
Similar to Case A, we need to compute the optimal Type-I plus Type-II error of testing
$$
    H_0: \{Y_i\}\mid \{\bfX_i\}, \{A_i\} \sim f_{\pi_{\underline\Delta}|\{\bfX_i\}, \{A_i\}}
    ~~~
    \textnormal{v.s.}
    ~~~
    H_1: \{Y_i\}\mid \{\bfX_i\}, \{A_i\} \sim f_{\pi_{\overline\Delta}|\{\bfX_i\}, \{A_i\}}.
$$
The likelihood function under $H_0$ is given by \eqref{eq:minimax_length_L0_caseA}, and the likelihood function under $H_1$ is given by
\begin{align*}
    L_1 & = \bigg(\prod_{i\in\Rct, A_i = 1} \frac{1}{\sqrt{2\pi \sigma_\rct^2}} \exp\{-(Y_i - \overline\Delta/2 - \bsgamma_\rct^\top \bfX_i)^2/(2\sigma_\rct^2)\}\bigg) \\
    & \qquad \times \bigg(\prod_{i\in\Rct, A_i = 0}\frac{1}{\sqrt{2\pi \sigma_\rct^2}} \exp\{-(Y_i + \overline\Delta/2 - \bsgamma_\rct^\top \bfX_i)^2/(2\sigma_\rct^2)\}\bigg) \\
    & \qquad \times \bigg(\prod_{i\in\Obs, A_i = 1} \frac{1}{\sqrt{2\pi \sigma_\obs^2}} \exp\{-(Y_i -  \bsgamma_\obs^\top \bfX_i)^2/(2\sigma_\obs^2)\}\bigg) \\
    & \qquad \times
    \bigg(\prod_{i\in\Obs, A_i = 0} \frac{1}{\sqrt{2\pi \sigma_\obs^2}} \exp\{-(Y_i -\bsgamma_\obs^\top \bfX_i)^2/(2\sigma_\obs^2)\}\bigg).
\end{align*}    
The optimal test rejects $H_0$ in favor of $H_1$ when $L_1/L_0\geq 1$, which is equivalent to
\begin{align*}
    -\sum_{i\in\calC, A_i = 1} \bigg(2(Y_i - \gamma_\rct^\top \bfX_i) - {\overline\Delta}/{2}\bigg) 
    + 
    \sum_{i\in\calC, A_i = 0} \bigg(2(Y_i - \gamma_\rct^\top \bfX_i) +\overline\Delta/2\bigg)  \leq 0.
\end{align*} 
Under $H_0$, the left-hand side above is distributed as
$
N\left(\overline\Delta n_\rct/2, 4n_\rct \sigma_\rct^2\right),
$
and thus the type-I error is given by
$$
    \bbP\bigg(N\left(\overline\Delta n_\rct/2, 4n_\rct \sigma_\rct^2\right) \leq 0\bigg) = \Phi\bigg(-\frac{\overline\Delta}{4 \sigma_\rct/\sqrt{n_\rct}}\bigg) \geq \Phi(-1/4),
$$
where the last inequality is by $\overline\Delta \leq \sigma_\rct/\sqrt{n_\rct}$. A nearly identical argument shows that the type-I error is also lower bounded by $\Phi(-1/4)$. Thus, 
$$
1 - TV(f_{\pi_{\underline\Delta}|\{\bfX_i\}, \{A_i\}}, f_{\pi_{\overline\Delta}|\{\bfX_i\}, \{A_i\}})] \geq 2\Phi(-1/4),
$$
and
$$
    \len^\star_\alpha(\calP_{\underline\Delta}, \calP_{\overline\Delta}) 
    \geq 2 \overline\Delta \cdot (\Phi(-1/4) - \alpha)\gtrsim \overline\Delta,
$$
where the last inequality is by $\alpha \leq \Phi(-1/4)-\ep$.

\paragraph{Case C. $\boldsymbol{\overline\Delta \leq 1/\sqrt{n_\rct/\sigma_\rct^2 + n_\obs/\sigma_\obs^2}}$.}
In this case, we let $\pi_{\underline\Delta}$ be a point mass at $\beta^\star = 0, \beta_\obs = 0$ and $\pi_{\overline\Delta}$ be a point mass at $\beta^\star = \frac{1}{\sqrt{n_\rct/\sigma_\rct^2 + n_\obs/\sigma_\obs^2}}, \beta_\obs = \frac{1}{\sqrt{n_\rct/\sigma_\rct^2 + n_\obs/\sigma_\obs^2}}$.
By Lemmas \ref{lemma:minimax_length} and \ref{lemma:tv_ineq}, we have
\begin{align*}
    \len^\star_\alpha(\calP_{\underline\Delta}, \calP_{\overline\Delta}) 
    & \geq \frac{1}{\sqrt{n_\rct/\sigma_\rct^2 + n_\obs/\sigma_\obs^2}} \cdot \bigg[\bigg( \bbE_{\{\bfX_i\}, \{A_i\}}[ 1 - TV(f_{\pi_{\underline\Delta}|\{\bfX_i\}, \{A_i\}}, f_{\pi_{\overline\Delta}|\{\bfX_i\}, \{A_i\}})] - 2\alpha\bigg) \lor 0\bigg].
\end{align*}
Similar to the previous two cases, we need to compute the optimal Type-I plus Type-II error of testing
$$
    H_0: \{Y_i\}\mid \{\bfX_i\}, \{A_i\} \sim f_{\pi_{\underline\Delta}|\{\bfX_i\}, \{A_i\}}
    ~~~
    \textnormal{v.s.}
    ~~~
    H_1: \{Y_i\}\mid \{\bfX_i\}, \{A_i\} \sim f_{\pi_{\overline\Delta}|\{\bfX_i\}, \{A_i\}}.
$$
The likelihood function under $H_0$ is given by \eqref{eq:minimax_length_L0_caseA}, and the likelihood function under $H_1$ is given by
\begin{align*}
    L_1 & = \bigg(\prod_{i\in\Rct, A_i = 1} \frac{1}{\sqrt{2\pi \sigma_\rct^2}} \exp\{-(Y_i - (2\sqrt{n_\rct/\sigma_\rct^2 + n_\obs/\sigma_\obs^2})^{-1} - \bsgamma_\rct^\top \bfX_i)^2/(2\sigma_\rct^2)\}\bigg) \\
    & \qquad \times \bigg(\prod_{i\in\Rct, A_i = 0}\frac{1}{\sqrt{2\pi \sigma_\rct^2}} \exp\{-(Y_i + (2\sqrt{n_\rct/\sigma_\rct^2 + n_\obs/\sigma_\obs^2})^{-1} - \bsgamma_\rct^\top \bfX_i)^2/(2\sigma_\rct^2)\}\bigg) \\
    & \qquad \times \bigg(\prod_{i\in\Obs, A_i = 1} \frac{1}{\sqrt{2\pi \sigma_\obs^2}} \exp\{-(Y_i - (2\sqrt{n_\rct/\sigma_\rct^2 + n_\obs/\sigma_\obs^2})^{-1} -  \bsgamma_\obs^\top \bfX_i)^2/(2\sigma_\obs^2)\}\bigg) \\
    & \qquad \times
    \bigg(\prod_{i\in\Obs, A_i = 0} \frac{1}{\sqrt{2\pi \sigma_\obs^2}} \exp\{-(Y_i + (2\sqrt{n_\rct/\sigma_\rct^2 + n_\obs/\sigma_\obs^2})^{-1} -\bsgamma_\obs^\top \bfX_i)^2/(2\sigma_\obs^2)\}\bigg).
\end{align*}
The optimal test rejects $H_0$ in favor of $H_1$ when $L_1/L_0\geq 1$. With some algebra, such a test rejects $H_0$ in favor of $H_1$ when
\begin{align*}
    & \frac{-1}{\sigma_\rct^2} \sum_{i\in\Rct, A_i = 1} \bigg(2(Y_i - \gamma_\rct^\top \bfX_i) - \frac{1}{2 \sqrt{n_\rct/\sigma_\rct^2 + n_\obs/\sigma_\obs^2}}\bigg)
    + 
    \frac{1}{\sigma_\rct^2} \sum_{i\in\Rct, A_i = 0} \bigg(2(Y_i - \gamma_\rct^\top \bfX_i) + \frac{1}{2 \sqrt{n_\rct/\sigma_\rct^2 + n_\obs/\sigma_\obs^2}}\bigg) \\
    & ~ -\frac{1}{\sigma_\obs^2} \sum_{i\in\Obs, A_i = 1} \bigg(2(Y_i - \gamma_\obs^\top \bfX_i) - \frac{1}{2 \sqrt{n_\rct/\sigma_\rct^2 + n_\obs/\sigma_\obs^2}}\bigg) 
    + \frac{1}{\sigma_\obs^2} \sum_{i\in\Obs, A_i = 0} \bigg(2(Y_i + \gamma_\obs^\top \bfX_i) - \frac{1}{2 \sqrt{n_\rct/\sigma_\rct^2 + n_\obs/\sigma_\obs^2}}\bigg)\\
    & \leq 0.
\end{align*}
The left-hand side above is distributed as
$$
    \frac{1}{\sigma_\rct^2} \cdot N\bigg(\frac{n_\rct}{2\sqrt{n_\rct/\sigma_\rct^2 + n_\obs/\sigma_\obs^2}}, 4 n_\rct \sigma_\rct^2\bigg) \Conv
    \frac{1}{\sigma_\obs^2} \cdot N\bigg(\frac{n_\obs}{2\sqrt{n_\rct/\sigma_\rct^2 + n_\obs/\sigma_\obs^2}}, 4 n_\obs \sigma_\obs^2\bigg),
$$
where $\conv$ denotes the convolution of random variables. Thus, the type-I error is given by
\begin{align*}
    & \bbP\bigg[ \frac{n_\rct/\sigma_\rct^2}{2\sqrt{n_\rct/\sigma_\rct^2 + n_\obs/\sigma_\obs^2}} +  \frac{n_\obs/\sigma_\obs^2}{2\sqrt{n_\rct/\sigma_\rct^2 + n_\obs/\sigma_\obs^2}} + \bigg(\frac{2\sqrt{n_\rct}}{\sigma_\rct} \cdot N(0, 1)\bigg)\Conv\bigg(\frac{2\sqrt{n_\obs}}{\sigma_\obs} \cdot N(0, 1)\bigg)\leq 0 \bigg]\\
    & = \bbP\bigg[ \frac{\sqrt{n_\rct/\sigma_\rct^2 + n_\obs/\sigma_\obs^2}}{2} + 2 \cdot N \bigg(0, \frac{n_\rct}{\sigma_\rct^2} + \frac{n_\obs}{\sigma_\obs^2}\bigg)\leq 0\bigg] \\
    & = \Phi(-1/4).
\end{align*}
A nearly identical argument shows that the type-I error is also lower bounded by $\Phi(-1/4)$. Thus, 
$$
    \len^\star_\alpha(\calP_{\underline\Delta}, \calP_{\overline\Delta}) 
    \geq 2 \frac{1}{\sqrt{n_\rct/\sigma_\rct^2 + n_\obs/\sigma_\obs^2}} \cdot (\Phi(-1/4) - \alpha)\gtrsim \frac{1}{\sqrt{n_\rct/\sigma_\rct^2 + n_\obs/\sigma_\obs^2}},
$$
where the last inequality is by $\alpha \leq \Phi(-1/4)-\ep$.

\paragraph{Finishing the proof.}
Combining the above three cases, we get 
$$
    \len^\star_\alpha(\calP_{\underline\Delta}, \calP_{\overline\Delta})  \gtrsim \frac{1}{\sqrt{n_\rct/\sigma_\rct^2 + n_\obs/\sigma_\obs^2}} \lor \bigg(\frac{\sigma_\rct}{\sqrt{n_\rct}} \land \overline\Delta\bigg) 
    \asymp \frac{1}{\sqrt{n_\rct/\sigma_\rct^2 + n_\obs/\sigma_\obs^2}} +  \frac{\sigma_\rct}{\sqrt{n_\rct}} \land \overline\Delta,
$$
and the proof is finished.

\newpage
\section{Additional Simulation Results}\label{appx:more_simulation}
\begin{table}[ht]
\centering
\caption{Ratios of the mean squared error of each estimator to that of the oracle estimator $\widehat{\beta}_\weight$ when $n_o = 50,000$. Estimators considered in the table are the RCT-only estimator $\widehat{\beta}_c$, the observational-data-only estimator $\widehat{\beta}_o$, and the anchored thresholding estimator $\widehat{\beta}_{\lambda}$.
}
\resizebox{\textwidth}{!}{
\begin{tabular}{ccccccccccc}
  
 &$\hat\beta_c$ & $\hat\beta_o$ &\multicolumn{2}{c}{$\hat\beta_{\lambda}$} & $\hat\beta_c$ & $\hat\beta_o$ &\multicolumn{2}{c}{$\hat\beta_{\lambda}$}\\
 \hline\hline
   \multicolumn{9}{c}{$\beta(X_1, X_2, X_3) = 2$} \\
   \hline
 \multirow{2}{*}{\begin{tabular}{c}b\end{tabular}} 
 &\multirow{2}{*}{\begin{tabular}{c}\textsf{IPW}\end{tabular}} 
 &\multirow{2}{*}{\begin{tabular}{c}\textsf{IPW}\end{tabular}} 
   & \multirow{2}{*}{\begin{tabular}{c}$\lambda_1=0.5$\end{tabular}}
   & \multirow{2}{*}{\begin{tabular}{c}$\lambda_1=0.6$\end{tabular}}
    &\multirow{2}{*}{\begin{tabular}{c}\textsf{AIPW}\end{tabular}} 
 &\multirow{2}{*}{\begin{tabular}{c}\textsf{AIPW}\end{tabular}} 
   & \multirow{2}{*}{\begin{tabular}{c}$\lambda_1=0.5$\end{tabular}}
   & \multirow{2}{*}{\begin{tabular}{c}$\lambda_1=0.6$\end{tabular}}
  \\  
  \\
  
   0.00 & 30.54 & 1.00 & 4.49 & 2.78 & 24.50 & 1.02 & 3.51 & 2.24 \\ 
   0.01 & 28.76 & 1.05 & 3.77 & 2.33 & 25.40 & 1.06 & 3.58 & 2.22 \\ 
  0.10 & 9.10 & 1.07 & 1.85 & 1.40 & 7.14 & 1.08 & 1.59 & 1.28 \\ 
   0.50 & 1.00 & 1.95 & 1.16 & 1.30 & 1.00 & 2.73 & 1.41 & 1.64 \\ 
   0.60 & 1.00 & 3.08 & 1.51 & 1.80 & 1.00 & 4.21 & 1.73 & 2.13 \\ 
   0.70 & 1.00 & 3.91 & 1.83 & 2.21 & 1.00 & 5.29 & 2.05 & 2.51 \\ 
   2.00 & 1.00 & 14.37 & 2.37 & 3.16 & 1.00 & 20.16 & 2.44 & 3.26 \\ 
  3.00 & 1.00 & 18.21 & 2.42 & 3.24 & 1.00 & 23.89 & 2.30 & 3.06 \\ 
  10.00 & 1.00 & 22.61 & 2.48 & 3.32 & 1.00 & 32.11 & 2.53 & 3.39 \\  
  \hline\hline
   \multicolumn{9}{c}{$\beta(X_1, X_2, X_3) = 2 - X_1 - X_2$} \\
   \hline

\multirow{2}{*}{\begin{tabular}{c}b\end{tabular}} 
 &\multirow{2}{*}{\begin{tabular}{c}\textsf{IPPW}\end{tabular}} 
 &\multirow{2}{*}{\begin{tabular}{c}\textsf{IPW}\end{tabular}} 
   & \multirow{2}{*}{\begin{tabular}{c}$\lambda_1=0.5$\end{tabular}}
   & \multirow{2}{*}{\begin{tabular}{c}$\lambda_1=0.6$\end{tabular}}
    &\multirow{2}{*}{\begin{tabular}{c}\textsf{AIPPW}\end{tabular}} 
 &\multirow{2}{*}{\begin{tabular}{c}\textsf{AIPW}\end{tabular}} 
   & \multirow{2}{*}{\begin{tabular}{c}$\lambda_1=0.5$\end{tabular}}
   & \multirow{2}{*}{\begin{tabular}{c}$\lambda_1=0.6$\end{tabular}}
  \\  \\  
   0.00 & 27.78 & 1.03 & 3.76 & 2.33 & 24.19 & 1.01 & 3.56 & 2.28 \\ 
   0.01 & 25.44 & 1.01 & 3.40 & 2.11 & 21.77 & 1.01 & 3.02 & 1.95 \\ 
   0.10 & 8.82 & 1.07 & 1.73 & 1.32 & 6.71 & 1.07 & 1.51 & 1.23 \\ 
   0.50 & 1.00 & 2.14 & 1.23 & 1.41 & 1.00 & 2.97 & 1.48 & 1.75 \\ 
   0.60 & 1.00 & 2.76 & 1.49 & 1.73 & 1.00 & 3.78 & 1.73 & 2.09 \\ 
   0.70 & 1.00 & 3.91 & 1.77 & 2.15 & 1.00 & 5.03 & 1.86 & 2.31 \\ 
   2.00 & 1.00 & 12.53 & 2.36 & 3.08 & 1.00 & 17.67 & 2.38 & 3.12 \\ 
   3.00 & 1.00 & 18.79 & 2.44 & 3.28 & 1.00 & 26.05 & 2.51 & 3.37 \\ 
   10.00 & 1.00 & 21.66 & 2.45 & 3.27 & 1.00 & 29.74 & 2.35 & 3.13 \\ 
\end{tabular}}
\end{table}

\begin{table}[ht]
\centering
\caption{Ratios of the mean squared error of each estimator to that of the oracle estimator $\widehat{\beta}_\weight$ when $n_o = 100,000$. Estimators considered in the table are the RCT-only estimator $\widehat{\beta}_c$, the observational-data-only estimator $\widehat{\beta}_o$, and the anchored thresholding estimator $\widehat{\beta}_{\lambda}$.
}
\resizebox{\textwidth}{!}{ 
\begin{tabular}{ccccccccccc}
 &$\hat\beta_c$ & $\hat\beta_o$ &\multicolumn{2}{c}{$\hat\beta_{\lambda}$} & $\hat\beta_c$ & $\hat\beta_o$ &\multicolumn{2}{c}{$\hat\beta_{\lambda}$}\\
 \hline\hline
   \multicolumn{9}{c}{$\beta(X_1, X_2, X_3) = 2$} \\
   \hline
 \multirow{2}{*}{\begin{tabular}{c}b\end{tabular}} 
 &\multirow{2}{*}{\begin{tabular}{c}\textsf{IPW}\end{tabular}} 
 &\multirow{2}{*}{\begin{tabular}{c}\textsf{IPW}\end{tabular}} 
   & \multirow{2}{*}{\begin{tabular}{c}$\lambda_1=0.5$\end{tabular}}
   & \multirow{2}{*}{\begin{tabular}{c}$\lambda_1=0.6$\end{tabular}}
    &\multirow{2}{*}{\begin{tabular}{c}\textsf{AIPW}\end{tabular}} 
 &\multirow{2}{*}{\begin{tabular}{c}\textsf{AIPW}\end{tabular}} 
   & \multirow{2}{*}{\begin{tabular}{c}$\lambda_1=0.5$\end{tabular}}
   & \multirow{2}{*}{\begin{tabular}{c}$\lambda_1=0.6$\end{tabular}}
  \\  \\
  
   0.00 & 28.75 & 1.04 & 3.70 & 2.26 & 24.32 & 1.04 & 3.46 & 2.23 \\ 
   0.01 & 30.02 & 1.05 & 3.99 & 2.39 & 25.30 & 1.05 & 3.66 & 2.33 \\ 
  0.10 & 8.65 & 1.08 & 1.62 & 1.25 & 6.37 & 1.09 & 1.40 & 1.17 \\ 
   0.50 & 1.00 & 2.20 & 1.21 & 1.41 & 1.00 & 3.09 & 1.52 & 1.79 \\ 
   0.60 & 1.00 & 2.77 & 1.42 & 1.68 & 1.00 & 3.76 & 1.66 & 2.02 \\ 
   0.70 & 1.00 & 3.59 & 1.67 & 2.00 & 1.00 & 4.90 & 1.87 & 2.31 \\ 
   2.00 & 1.00 & 14.05 & 2.30 & 3.07 & 1.00 & 19.78 & 2.34 & 3.13 \\ 
   3.00 & 1.00 & 16.92 & 2.22 & 2.96 & 1.00 & 22.69 & 2.20 & 2.92 \\ 
   10.00 & 1.00 & 21.80 & 2.48 & 3.30 & 1.00 & 30.30 & 2.45 & 3.26 \\  \hline\hline 
   \multicolumn{9}{c}{$\beta(X_1, X_2, X_3) = 2 - X_1 - X_2$} \\ \hline

\multirow{2}{*}{\begin{tabular}{c}b\end{tabular}} 
 &\multirow{2}{*}{\begin{tabular}{c}\textsf{IPPW}\end{tabular}} 
 &\multirow{2}{*}{\begin{tabular}{c}\textsf{IPW}\end{tabular}} 
   & \multirow{2}{*}{\begin{tabular}{c}$\lambda_1=0.5$\end{tabular}}
   & \multirow{2}{*}{\begin{tabular}{c}$\lambda_1=0.6$\end{tabular}}
    &\multirow{2}{*}{\begin{tabular}{c}\textsf{AIPPW}\end{tabular}} 
 &\multirow{2}{*}{\begin{tabular}{c}\textsf{AIPW}\end{tabular}} 
   & \multirow{2}{*}{\begin{tabular}{c}$\lambda_1=0.5$\end{tabular}}
   & \multirow{2}{*}{\begin{tabular}{c}$\lambda_1=0.6$\end{tabular}}
  \\  \\  
    0.00 & 28.75 & 1.04 & 3.70 & 2.26 & 24.32 & 1.04 & 3.46 & 2.23 \\ 
   0.01 & 30.02 & 1.05 & 3.99 & 2.39 & 25.30 & 1.05 & 3.66 & 2.33 \\ 
   0.10 & 8.65 & 1.08 & 1.62 & 1.25 & 6.37 & 1.09 & 1.40 & 1.17 \\ 
   0.50 & 1.00 & 2.20 & 1.21 & 1.41 & 1.00 & 3.09 & 1.52 & 1.79 \\ 
   0.60 & 1.00 & 2.77 & 1.42 & 1.68 & 1.00 & 3.76 & 1.66 & 2.02 \\ 
   0.70 & 1.00 & 3.59 & 1.67 & 2.00 & 1.00 & 4.90 & 1.87 & 2.31 \\ 
   2.00 & 1.00 & 14.05 & 2.30 & 3.07 & 1.00 & 19.78 & 2.34 & 3.13 \\ 
   3.00 & 1.00 & 16.92 & 2.22 & 2.96 & 1.00 & 22.69 & 2.20 & 2.92 \\ 
  10.00 & 1.00 & 21.80 & 2.48 & 3.30 & 1.00 & 30.30 & 2.45 & 3.26 \\
\end{tabular}}
\end{table}

\end{appendices}

\end{document}